\documentclass[11pt]{article}

\usepackage{times,epsfig,epsf,psfrag,array,amsmath,amssymb,graphicx,graphics,verbatim}
\usepackage[algo2e,boxruled,lined,linesnumbered]{algorithm2e}
\usepackage{url}
\usepackage{comment}
\newcommand{\TT}{{\cal T}}
\newcommand{\PP}{{\cal P}}
\newcommand{\DD}{{\cal D}}

\newcommand{\TB}{{\cal T_B}}

\newcommand{\ignore}[1]{{}}

\newenvironment{proof}{\par\noindent{\bf Proof:}}{\mbox{}\hfill$\Box$\\}

\newtheorem{definition}{Definition}[section]

\newtheorem{theorem}{Theorem}[section]

\newtheorem{lemma}{Lemma}[section]

\begin{document}
\date{}

\title{Dynamic DFS in Undirected Graphs: breaking the $O(m)$ barrier
\thanks{The preliminary version of the paper appeared in SODA 2016.}}

%
\newcommand{\specificthanks}[1]{$\ddagger$}

\author{Surender Baswana	
\textsuperscript{\specificthanks{2}}
\thanks{This research was partially supported 
by {\em UGC-ISF} (the University Grants Commission of India \& Israel Science Foundation) and
{\em IMPECS} (the Indo-German Max Planck Center for Computer Science) .}
\and Shreejit Ray Chaudhury 
\thanks{Dept. of CSE, I.I.T. Kanpur, India (www.cse.iitk.ac.in), 
email: \{sbaswana,keerti,shahbazk\}@cse.iitk.ac.in, shreejit.1@gmail.com.}
\and Keerti Choudhary 
\textsuperscript{\specificthanks{2}}
\thanks{This research was partially supported by Google India under the Google 
India PhD Fellowship Award.}
\and Shahbaz Khan 
\textsuperscript{\specificthanks{2}}\footnotemark[4]
}


\maketitle

Depth first search (DFS) tree is a fundamental data structure for solving various problems in graphs. It is well known that it takes $O(m+n)$ time to build a DFS tree for a given undirected graph $G=(V,E)$ on $n$ vertices and $m$ edges. We address the problem of maintaining a DFS tree when the graph is undergoing {\em updates} (insertion and deletion of vertices or edges). We present the following results for this problem.

\begin{enumerate}
\item {\em Fault tolerant DFS tree}:
There exists a data structure of size $\tilde{O}(m)$ \footnote{$\tilde{O}()$ hides the poly-logarithmic factors.} such that given any set ${\cal F}$ of failed vertices or edges, a DFS tree of the graph $G\setminus {\cal F}$ can be reported in $\tilde{O}(n|{\cal F}|)$ time. 

\item {\em Fully dynamic DFS tree}:
There exists a fully dynamic algorithm for maintaining a DFS tree that takes worst case $\tilde{O}(\sqrt{mn})$ time per update  for any arbitrary online sequence of updates.

\item {\em Incremental DFS tree}:
There exists an incremental algorithm for maintaining a DFS tree that takes worst case $\tilde{O}(n)$ time per update for any arbitrary online sequence of edge insertion.
\end{enumerate}

These are the first $o(m)$ worst case time results for maintaining a DFS tree in a dynamic environment. Moreover, our fully dynamic algorithm provides, in a seamless manner, the first deterministic algorithm with $O(1)$ query time and $o(m)$ worst case update time for connectivity, biconnectivity, and 2-edge connectivity in the dynamic subgraph model. 


\section{Introduction}
Depth First Search (DFS) is a well known graph traversal technique. Right from the seminal work of Tarjan \cite{Tarjan72}, DFS traversal 
has played the central role in the design of efficient algorithms for many fundamental graph problems, namely, biconnected components~\cite{Tarjan72}, 
strongly connected components~\cite{Tarjan72}, topological sorting,
bipartite matching \cite{HopcroftK73}, dominators in directed graph \cite{Tarjan74} and planarity testing \cite{HopcroftT74}.
Interestingly, the role of DFS traversal is not confined to merely the design of efficient algorithms. 
For example, consider the classical result of Erd{\H o}s and R{\' e}nyi \cite{ErdosR60} for the 
phase transition phenomena in random graphs. There exist many proofs of this result which are intricate and based on highly sophisticated probability tools. However, recently, Krivelevich and Sudakov \cite{KrivelevichS13}  designed a truly simple, short, and elegant proof for this result based on the insights from a DFS traversal in a graph.

Let $G=(V,E)$ be an undirected connected graph on $n$ 
vertices and $m$ 
edges. DFS traversal of $G$ starting from any vertex $r\in V$ 
produces a rooted spanning tree, 
called a DFS tree with $r$ as its root. 
It takes $O(m+n)$ time to perform a DFS traversal and generate a DFS tree. 
Given any rooted spanning tree of graph $G$, all non-tree edges of the graph
can be 
classified into two categories, namely, back edges and cross edges
as follows. A non-tree edge is called a {\it back edge} if one of its endpoints is an ancestor of the other in the tree. Otherwise, it is called a
{\it cross edge}. A necessary and sufficient condition for any rooted 
spanning tree to be a DFS tree is that every non-tree edge is a back edge.
Thus, it can be seen that many DFS trees are possible for any given graph. 
However, if the traversal of the graph is performed according to the order specified by the adjacency lists of the graph, the resulting DFS tree will
be unique. The ordered DFS tree problem
is to compute the order in which the vertices get visited when the traversal
is performed strictly according to the adjacency lists.

Most of the graph applications in real world deal with graphs that keep 
changing with time. These changes/updates can be in the form of insertion 
or deletion of vertices or edges. An algorithmic graph problem is modeled in a dynamic environment as follows. There is an online sequence of 
updates on the graph, and the objective is to update the solution of the problem efficiently after each update. 
In particular, the time taken to update the solution has to be much smaller than that of the best static algorithm for the problem. 
In the last two decades, many elegant dynamic algorithms have been designed for various graph problems such as connectivity \cite{EppsteinGIN97,HenzingerK99,HolmLT01,KapronKM13},
reachability \cite{RodittyZ08,Sankowski04},
shortest path \cite{DemetrescuI04,RodittyZ12}, spanners \cite{BaswanaKS12,GottliebR08,Roditty12}, and min-cut \cite{Thorup07}.
Another, and more restricted, variant of a dynamic environment is the fault tolerant environment. Here the aim is to build a compact
data structure for a given problem, that is resilient to failure of vertices/edges, and can efficiently report the 
solution of the problem for any given set of failures. 
There has been a lot of work in the last two decades on fault tolerant algorithms for connectivity 
\cite{ChanPR08,Duan10,FrigioniI00}, shortest paths \cite{BaswanaK13,ChechikLPR12,DemetrescuTCR08}, and spanners
\cite{BraunschvigCPS15,ChechikLPR10}.

A dynamic graph algorithm is said to be \textit{fully dynamic} if it handles both 
insertion as well as deletion updates. A partially dynamic algorithm is said
to be \textit{incremental} or \textit{decremental} if it handles only insertion or only deletion updates respectively. 
In this paper, we address the problem of maintaining a DFS tree efficiently in any dynamic environment.  

\subsection{Existing results on dynamic DFS}
In spite of the simplicity of a DFS tree, designing any efficient parallel 
or dynamic algorithm for a DFS tree has turned out to be quite challenging. Reif \cite{Reif85} 
showed that the ordered DFS tree problem is a $P$-Complete problem. 
For many years, this result seemed to imply that the general DFS tree problem,
that is, the computation of any DFS tree, is also inherently sequential. 
However, Aggarwal and Anderson \cite{AggarwalA88} proved that the general
DFS tree problem is in {\em RNC} by designing a parallel randomized algorithm that takes 
$O(\log^3 n)$ expected time. 
Further, the fastest parallel deterministic algorithm for general DFS tree still 
takes $O(\sqrt{n})$ time \cite{AggarwalAK90,GoldbergPV88}.  
Whether the general DFS tree problem is in {\em NC} for directed (or undirected) graphs is still a 
long standing open problem.

Reif \cite{Reif87} and later Miltersen et al. \cite{MiltersenSVT94} proved 
that $P$-Completeness of a problem also implies hardness of the problem in 
the dynamic setting. The work of Miltersen et al. \cite{MiltersenSVT94} 
shows that if the ordered DFS tree is updateable in $O({\mathbf{polylog}}(n))$ time, 
then the solution of every problem in class $P$ is updateable in 
$O({\mathbf{polylog}}(n))$ time. In other words, maintaining the ordered DFS 
tree is indeed the hardest among all the problems in class $P$.
In our view, this hardness result, which is actually for only the ordered
DFS tree problem, has proved to be quite discouraging for the researchers 
working in the area of dynamic algorithms.
This is evident from the fact that for all the static graph problems that
were solved using DFS traversal in the 1970's, none of their dynamic counterparts
used a dynamic DFS tree \cite{HenzingerK99,HolmLT01,KapronKM13,RodittyZ08,Chan06,ChanPR08,Duan10}.

Apart from the hardness of the ordered DFS tree problem in dynamic environment, 
very little progress has been achieved even for the problem
of maintaining any DFS tree. 
Franciosa et al. \cite{FranciosaGN97} designed an incremental algorithm for
a DFS tree in a directed acyclic graph (DAG). 
For any arbitrary sequence of edge insertions, this algorithm takes $O(mn)$ total time to 
maintain a DFS tree from a given source. 
Recently, Baswana and Choudhary \cite{BaswanaC15} designed a decremental algorithm
for a DFS tree in a DAG that requires expected $O(mn\log n)$ total time.
For undirected graphs, recently 
Baswana and Khan \cite{BaswanaK14} designed an incremental algorithm for maintaining a DFS tree
requiring $O(n^2)$ total time. 
These algorithms are the only results known for the dynamic DFS tree problem. 
Moreover, none of these existing algorithms, though designed for only a partially dynamic environment, 
achieves a worst case bound of $o(m)$ on the update time.
%
Furthermore, none of these results proves that general DFS is not as hard as ordered DFS in the dynamic environment. 
This is because the speculations of having to incur a complete recomputation in the worst case after an update 
is not disproved by amortized bounds resulting in the perceived $O(m)$ {\em barrier} for general DFS as well.  
So the following intriguing questions remained unanswered till date:
\begin{itemize}
\item
Does there exist any fully dynamic algorithm for maintaining
a DFS tree?
\item 
Is it possible to achieve worst case $o(m)$ update time for maintaining a DFS tree in a dynamic environment?
\end{itemize}

Not only do we answer these open questions affirmatively for undirected graphs, 
we also use our dynamic algorithm for DFS tree to provide efficient 
solutions for a couple of well studied dynamic graph problems. 
Moreover, our results also handle vertex updates which are generally considered harder than
edge updates.
Furthermore, our results finally prove that general DFS is indeed not as hard as ordered DFS in the dynamic setting
as was the case in parallel setting.
%
%

\subsection{Our results}
%
We consider a generalized notion of updates wherein an update could be either
insertion/deletion of a vertex or insertion/deletion of an edge. For any set
$U$ of such updates, let $G+U$ denote the graph obtained after performing 
the updates $U$ on the graph $G$. Our main result can be succinctly
described in the following theorem.
\begin{theorem} An undirected graph can be preprocessed to build a data 
structure of $O(m \log n)$ size such that for any set $U$ of $k\leq n$ updates, a DFS tree of $G+U$ can be reported in $O(nk \log^4 n)$ 
time.
\label{main-result}
\end{theorem}

With this result at the core, we easily obtain the following 
results for dynamic DFS tree in an undirected graph.
\begin{enumerate}
\item {\em Fault Tolerant DFS tree}:

Given any set of $k$ failed vertices or edges, 
we can report a DFS tree 
for the resulting graph in $O(nk \log^4 n)$ time.

%
%

\item {\em Fully Dynamic DFS tree}:

Given any arbitrary online sequence of vertex or edge updates, we can maintain a DFS tree in $O(\sqrt{mn} \log^{2.5} n)$ worst case time per update. 


\item {\em Incremental DFS tree}:

Given any arbitrary online sequence of edge insertions, we can maintain
a DFS tree in $O(n \log^{3} n)$ worst case time per edge insertion. 


\end{enumerate} 

These are the first $o(m)$ worst case update time algorithms for maintaining a DFS tree in a dynamic environment. 
%
Recently, there has been significant work
\cite{AbboudW14, HenzingerKNS15} on establishing conditional lower bounds on the time complexity of various dynamic graph problems.
A simple reduction from [1], based on the Strong Exponential Time Hypothesis (SETH), implies a conditional lower bound of  ${\Omega}(n)$ on the update time of any fully dynamic algorithm for a DFS tree under vertex updates. 
We also present an unconditional lower bound of $\Omega(n)$ for
maintaining a fully dynamic DFS tree explicitly under edge updates.



\subsection{Applications of Fully Dynamic DFS}

In the static setting, a DFS tree can be easily used to answer connectivity, 2-edge connectivity and biconnectivity queries.
Our fully dynamic algorithm for DFS tree thus seamlessly solves these problems for both vertex and edge updates.
Further, our result gives the first deterministic algorithm with $O(1)$ query time and $o(m)$ worst case update time 
for several well studied variants of these problems in the dynamic setting.
These problems include dynamic subgraph connectivity \cite{ChanPR08,Duan10,EppsteinGIN97,Frederickson85,HolmLT01,KapronKM13} 
and vertex update versions of dynamic biconnectivity \cite{Henzinger00,Henzinger95,HolmLT01} and dynamic 2-edge connectivity \cite{HolmLT01,EppsteinGIN97,Frederickson85}. 
The existing results offer different trade-offs between the update time and the query time, and differ on the types 
(amortized or worst case) of update time and the types (deterministic or randomized) of query time. 
Our algorithm, in particular, improves the deterministic worst 
case bounds for these problems, thus demonstrating 
the relevance of
DFS trees in solving dynamic graph problems.

\subsection{Main Idea}
Let $T$ be a DFS tree of $G$. To compute a DFS tree of $G+U$ for a given set
$U$ of updates, the main idea is to make use of the original tree $T$ itself. We preprocess the
graph $G$ using tree $T$ to build a data structure ${\cal D}$. In order
to achieve $o(m)$ update time, our algorithm makes use of ${\cal D}$ to create a {\it reduced} adjacency list
for each vertex such that performing DFS traversal using these lists gives a DFS tree for  $G+U$. 
In fact, these reduced adjacency lists are generated on the fly and are guaranteed to have only $\tilde{O}(n|U|)$ edges. 

We now give an outline of the paper.
In section \ref{sec:prelim}, we describe various notations used throughout the paper.
Section \ref{sec:Overview1} describes an algorithm to report the DFS tree after a single update in the graph.
The details of the required data structure $\DD$ are described in Section \ref{sec:data-structure}.
Then in Section \ref{sec:Overview2}, we provide an overview of our algorithm for handling 
multiple updates, highlighting the main intuition behind our approach.
Our main algorithm (Theorem \ref{main-result}) that reports a DFS tree after 
any set of updates in the graph is described in Section \ref{sec:algo}.
In Section \ref{sec:fullyDyn} we convert this algorithm to fully dynamic and incremental algorithms 
for maintaining a DFS tree using the \textit{overlapped periodic rebuilding} technique. 
Finally, in Section \ref{sec:appn} and Section \ref{sec:LowerBounds} we describe the applications and 
lower bounds of dynamic DFS trees. 

\section{Preliminaries}
\label{sec:prelim}

Let $U$ be any given set of updates. We add a dummy 
vertex $r$ to the given graph in the beginning and connect it to all the
vertices. Our algorithm starts with any arbitrary DFS tree $T$ rooted at 
$r$ in the augmented graph and it maintains a DFS tree rooted at $r$ at
each stage. It can be observed easily that each subtree rooted at any 
child of $r$ is a DFS tree of a connected component of the graph $G+U$.
The following notations will be used throughout the paper.

\begin{itemize}
\itemsep0em 
\item  $T(x):$ The subtree of $T$ rooted at vertex $x$.
\item  $path(x,y):$ Path from the vertex $x$ to the vertex $y$ in $T$.
\item  $dist_T(x,y):$ The number of edges on the path from $x$ to $y$ in $T$.
\item  $LCA(x,y):$ The lowest common ancestor of $x$ and $y$ in tree $T$.
\item  $N(w):$ 
The adjacency list of vertex $w$ in the graph $G+U$.
\item  $L(w):$ 
The reduced adjacency list of vertex $w$ in the graph $G+U$.
\item $T^*:$~ The DFS tree rooted at $r$ computed by our algorithm for
the graph $G+U$.
\item $par(w):$~ Parent of $w$ in $T^*$.
\end{itemize}

A subtree $T'$ is said to be {\em hanging} from a path $p$ if the root $r'$ of $T'$ is a child of some vertex on the path $p$
and $r'$ does not belong to the path $p$. 
Unless stated otherwise, every reference to a path refers to an
ancestor-descendant path defined as follows:

\begin{definition}[Ancestor-descendant path]
A path $p$ in a DFS tree $T$ is said to be ancestor-descendant path 
if its endpoints have ancestor-descendant relationship in $T$.
\label{definition:anc-des}
\end{definition}

We now state the operations supported by the data structure $\DD$ (complete details of $\DD$ are in Section \ref{sec:data-structure}). 
Let $U$ below refer to a set of updates that consists of vertex and edge deletions only.
For any three vertices $w,x,y\in T$, where $path(x,y)$ is an 
ancestor-descendant path in $T$, the following two queries can be answered using $\DD$ in $O(\log^3 n)$ time.
\begin{enumerate}
\item $Query(w,x,y):$ 
among all the edges from $w$ that are incident on $path(x,y)$ 
in $G+U$, return an edge that is incident nearest to $x$ on $path(x,y)$.

\item $Query(T(w),x,y):$ 
among all the edges from $T(w)$ that are incident on $path(x,y)$ 
in $G+U$, return an edge that is incident nearest to $x$ on $path(x,y)$.
\end{enumerate}

We now describe an important property of a DFS traversal that will be crucially 
used in our algorithm.  



\subsection{Properties of a DFS tree}
DFS traversal has the following flexibility : when the traversal reaches  a vertex, say $v$, the next vertex 
to be traversed can be {\em any} unvisited neighbor of $v$. 
In order to compute a DFS tree for $G+U$ efficiently, our algorithm exploits this flexibility,
the original DFS tree $T$, and the following property of DFS traversal. 

\begin{figure}[!hb]
\centering
\includegraphics[width=0.35\linewidth]{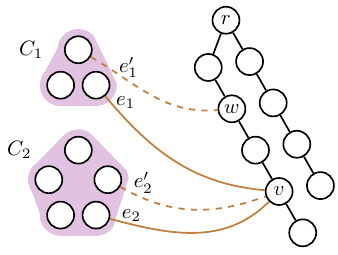}
\caption{Edges $e_1'$ as well as $e_2'$ can be ignored during the DFS traversal.}
\label{fig:component-property}
\end{figure}

\begin{lemma}[Components Property]
\label{lemma:L()}
Let $T^*$ be the partially grown DFS tree and $v$ be the vertex currently being visited. 
Let $C$ be any connected component in the subgraph induced by the unvisited vertices.  Suppose two edges $e$ and $e'$ from $C$
are incident respectively on $v$ and some ancestor (not necessarily proper) $w$ of $v$ in $T^*$. Then it is sufficient to 
consider only $e$ during the rest of the DFS traversal, i.e., the edge $e'$ need not  be scanned. (Refer to Figure \ref{fig:component-property}).
\end{lemma}

Skipping $e'$ during the DFS traversal, as stated in the components property, is justified because $e'$ will appear
as a back edge in the resulting DFS tree. A similar property describing the \textit{inessential} edges of a DFS trees 
was used by Smith~\cite{Smith86} for computing a DFS tree of a planar graph in the parallel setting.
In order to highlight the importance of the components property, and to motivate the requirement of data structure $\cal D$,
we first consider a simpler case which deals with reporting a DFS tree after a single update in the graph.

\section{Handling a single update}
\label{sec:Overview1}

Consider the failure of a single edge $(b,f)$ (refer to Figure \ref{figure:overview-failure-eg} (i)).
Exploiting the flexibility of DFS traversal, we can assume a stage in the DFS traversal of $G\backslash \{(b,f)\}$ where
the partial DFS tree $T^*$ is $T\backslash T(f)$ and  vertex $b$ is currently being visited. 
Thus, the unvisited graph is a single connected component containing the vertices of $T(f)$. 
Now, according to the components property we need to process only the lowest edge from $T(f)$ to $path(b,r)$ 
($(k,b)$ in Figure \ref{figure:overview-failure-eg} (ii)). 
Hence, the DFS traversal enters this component using the edge $(k,b)$ and performs a traversal of the subgraph induced by the vertices of $T(f)$.
The resulting DFS tree of this subgraph would now be rooted at $k$. 
Rebuilding the DFS tree after the failure of edge $(b,f)$ thus reduces to 
finding the lowest edge from $T(f)$ to $path(e,r)$, 
and then rerooting a subtree $T(f)$ of $T$ at the new root $k$. 
We now describe how this rerooting can be performed in $\tilde{O}(n)$ time
in the following section.

\begin{figure}[!ht]
\centering
\includegraphics[width=\linewidth]{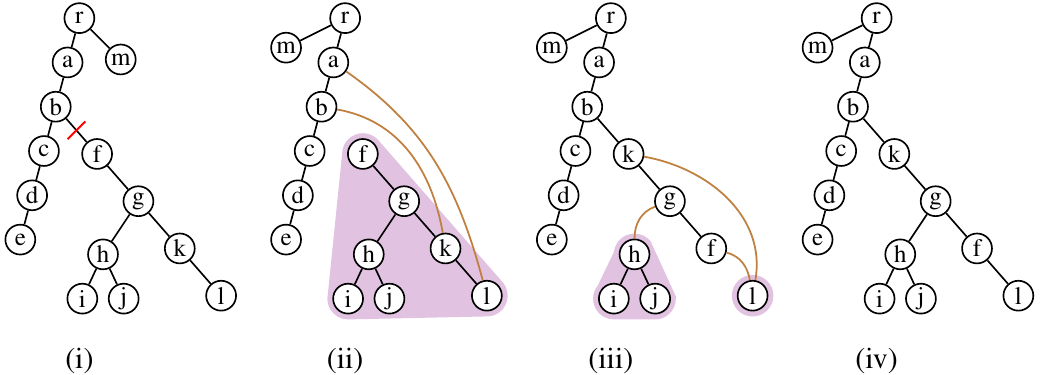}
\caption{ (i) Failure of edge $(b,f)$. 
(ii) Partial DFS tree $T^*$ with unvisited graph $T(f)$, component property allows us to neglect $(a,l)$.
(iii) Augmented $path(k,f)$ to $T^*$, the components property allows us to neglect $(l,k)$.
(iv) Final DFS tree of $G\backslash \{(b,f)\}$.
}
\label{figure:overview-failure-eg}
\end{figure}

\subsection{Rerooting a DFS tree}
Given a DFS tree $T$ originally rooted at $r_0$ and a vertex $r'$, 
the aim is to compute a DFS tree of the graph that is rooted at $r'$. 
Note that any subtree $T(x)$ of the DFS tree $T$ is also a DFS tree of the subgraph
induced by the vertices of $T(x)$. Hence, the same procedure can be 
applied to reroot a subtree $T(x)$ of the DFS tree $T$. 
Thus, in general our aim is to reroot $T(r_0)$ at a new root $r'\in T(r_0)$ 
(see Figure \ref{figure:overview-failure-eg} (ii), where the subtree 
$T(f)$ would be rerooted at its new root $k$).

\begin{figure*}[!ht]
\centering
\begin{procedure}[H]
\BlankLine
\ForEach(\tcc*[f]{$a=par(b)$ in original tree $T(r_0)$.}){$(a,b)$ on $path(r_0,r')$}
	{
	$par(a)\leftarrow b$\;
	\ForEach{child $c$ of $b$ not on $path(r_0,r')$}
		{
		$(u,v)\leftarrow Query(T(c),r_0,b)$
		\tcc*{where $u\in path(r_0,r')$ and $v\in T(c)$.} 
		\If{$(u,v)$ is non-null}
			{
			$Reroot(T(c),v)$\;
			$par(v)\leftarrow u$\;
			}
		}	
}
\caption{Reroot(\text{$T(r_0),r'$}): Reroots the subtree $T(r_0)$ of $T$ to be rooted at the vertex $r'\in T(r_0)$.}
\label{alg:re-root}
\end{procedure}
\caption{The recursive algorithm to reroot a DFS tree $T(r_0)$ at the new root $r'$.}
\end{figure*}

Our algorithm (refer to Procedure \ref{alg:re-root}) essentially performs 
the DFS traversal (exploiting the flexibility of DFS) 
in such a way
that components of the unvisited graph can be easily identified.
The components property can then be applied to each such component, processing only $O(n)$ edges 
to compute the rerooted DFS tree.
The DFS traversal first visits the path from $r'$ to the root of tree $T(r_0)$. 
This reverses $path(r_0,r')$ in the new DFS tree $T^*$ as now $r'$ would be an ancestor of $r_0$ (see Figure \ref{figure:overview-failure-eg} (iii)).
Now, each subtree hanging from $path(r',r_0)$ in $T$ forms a component of the unvisited graph. 
This is because the presence of any edge between these subtrees would imply a cross edge in the original DFS tree. 
Using the components property we know that for each of these subtrees, say $T_i$,  
we only need to process the lowest edge from $T_i$ on the new path from $r'$ to $r_0$ in $T^*$.
Since $path(r',r_0)$ is reversed in $T^*$, it is equivalent to processing the highest edge 
$e_i$ from $T_i$ to the $path(r_0,r')$ in $T$. Recall that this query can be answered by our data structure $\cal D$
in $O(\log^3 n)$ time (refer to Section \ref{sec:prelim}). Now, let $v_i$ be the end vertex of $e_i$ in $T_i$.
The DFS traversal will thus visit the component induced by the vertices of $T_i$ through $e_i$,
and produces its DFS tree that is rooted at $v_i$. This rerooting can be performed by invoking the rerooting procedure
recursively on the subtree $T_i$ with the new root $v_i$.


We now analyze the total time required by Procedure \ref{alg:re-root} to reroot a subtree $T'$ of the DFS tree $T$. 
The total time taken by our algorithm is proportional to the number of edges processed by the algorithm.
These edges include the {\em tree edges} that were a part of the original tree $T'$ and 
the {\em added edges} that are returned by the data structure $\cal D$.
Clearly, the number of tree edges in $T'$ are $O(|T'|)$. Also, since the added edges eventually become 
a part of the new DFS tree $T^*$, they too are bounded by the size of the tree $T'$. 
Further, the data structure $\cal D$ takes $O(\log^3 n)$ time to report each added edge. 
Hence the total time taken by our algorithm to rebuild $T'$ is $O(|T'|\log^3 n)$ time. 
Since $\cal D$ can be built in $O(m\log n)$ time (refer to Theorem \ref{thm:DS} in the Section \ref{sec:data-structure}),
we have the following theorem.  

%

\begin{theorem}
An undirected graph can be preprocessed to build a data structure in $O(m \log n)$ time, such that 
any subtree $T'$ of the DFS tree can be rerooted at any vertex in $T'$, in $O(|T'| \log^3 n)$ time.
\label{thm:reroot}
\end{theorem}

We now formally describe how rebuilding a DFS tree after an update can be reduced 
to this simple rerooting procedure (see Figure \ref{figure:overview-updates}). 
\begin{enumerate}
\item \textbf{Deletion of an edge $(u,v)$:}\\
In case $(u,v)$ is a back edge in $T$, simply delete it from the graph.
Otherwise, let $u = par(v)$ in $T$. 
The algorithm finds the lowest edge $(u',v')$ on the $path(u,r)$ from $T(v)$, where $v'\in T(v)$. 
The subtree $T(v)$ is then rerooted at its new root $v'$ and hanged from $u'$ using $(u',v')$ in the final tree $T^*$.
\item \textbf{Insertion of an edge $(u,v)$:}\\
In case $(u,v)$ is a back edge, simply insert it in the graph.
Otherwise, let $w$ be the LCA of $u$ and $v$ in $T$ and $v'$ be the child of $w$ such that $v\in T(v')$.
The subtree $T(v')$ is then rerooted at its new root $v$ and hanged from $u$ using $(u,v)$ in the final tree $T^*$.
\item \textbf{Deletion of a vertex $u$:}\\
Let $v_1,...,v_c$ be the children of $u$ in $T$. 
For each subtree $T(v_i)$, the algorithm finds the lowest edge $(u'_i,v'_i)$ on the $path(par(u),r)$ from $T(v_i)$, where $v'_i\in T(v_i)$. 
Each subtree $T(v_i)$ is then rerooted at its new root $v'_i$ and hanged from $u'_i$ using $(u'_i,v'_i)$ in the final tree $T^*$.
\item \textbf{Insertion of a vertex $u$:}\\
Let $v_1,...,v_c$ be the neighbors of $u$ in the graph. 
Make $u$ a child of some $v_j$ in $T^*$.
For each $v_i$, such that $v_i\notin path(v_j,r)$, let $T(v'_i)$ be the subtree hanging from $path(v_j,r)$
such that $v_i\in T(v'_i)$. 
Each subtree $T(v'_i)$ is then rerooted at its new root $v_i$ and 
hanged from $u$ using $(u,v_i)$ in the final tree $T^*$.
\end{enumerate}


\begin{figure}[!ht]
\centering
\includegraphics[width=\linewidth]{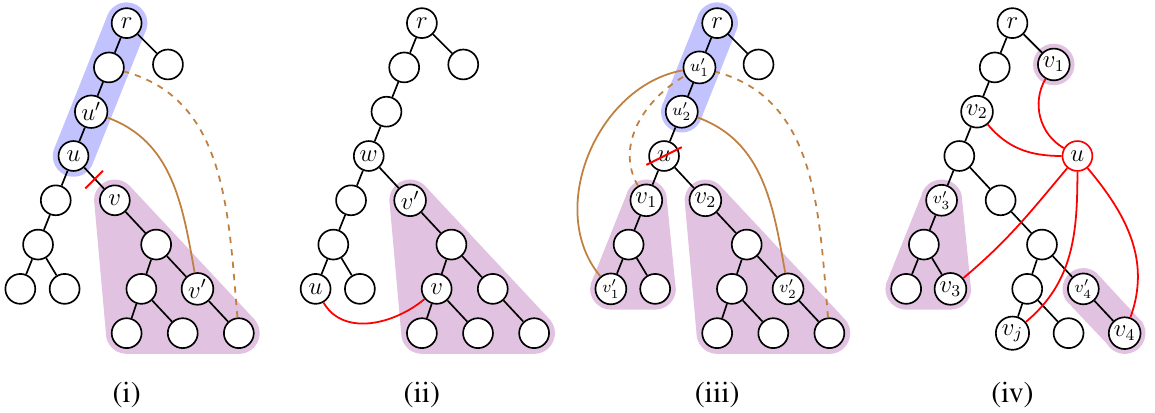}
\caption{Updating the DFS tree after a single update: (i) deletion of an edge, (ii) insertion of an edge, 
(iii) deletion of a vertex, and (iv) insertion of a vertex. 
The reduction algorithm reroots the marked subtrees (shown in violet) and hangs it from the 
inserted edge (in case of insertion) or the lowest edge (in case of deletion) on the marked path 
(shown in blue) from the marked subtree.
}
\label{figure:overview-updates}
\end{figure}

In case of vertex updates, multiple subtrees may be rerooted by the algorithm. Let these subtrees be $T_1,...,T_c$.
Thus, the total time taken by our algorithm is equal to the time taken to reroot the subtrees $T_1,...,T_c$.  
Using Theorem \ref{thm:reroot}, we know that a subtree $T'$ can be rerooted in $\tilde{O}(|T'|)$ time.
Since these subtrees are disjoint, the total time taken by our algorithm to build the resulting DFS tree is 
$\tilde{O}(|T_1|+...+|T_c|) = \tilde{O}(n)$. 
Thus, we have the following theorem.

\begin{theorem}
An undirected graph can be preprocessed to build a data structure in $O(m \log n)$ time such that 
after a single update in the graph, the DFS tree can be reported in $O(n \log^3 n)$ time.
\label{thm:DFS_reroot}
\end{theorem}

\section{Data Structure}
\label{sec:data-structure}
The efficiency of our algorithm heavily relies  on the
data structure ${\cal D}$. For any three vertices $w,x,y\in T$, where $path(x,y)$ is an 
ancestor-descendant path in $T$, we need to answer the following two kinds of queries. 
\begin{enumerate}
\item $Query(w,x,y):$ 
among all the edges from $w$ that are incident on $path(x,y)$ 
in $G+U$, return an edge that is incident nearest to $x$ on $path(x,y)$.

\item $Query(T(w),x,y):$ 
among all the edges from $T(w)$ that are incident on $path(x,y)$ 
in $G+U$, return an edge that is incident nearest to $x$ on $path(x,y)$.
\end{enumerate}

We now describe construction of the data structure $\cal D$. It employs
a combination of two well known techniques, 
namely, heavy-light decomposition \cite{SleatorT83} and suitable augmentation of a binary tree (segment tree)
as follows. 

\begin{enumerate}
\item Perform a preorder traversal of tree $T$ with the following restriction: Upon visiting a vertex $v\in T$, the child
of $v$ that is visited first is the one storing the largest subtree. 
Let ${\cal L}$ be the list of vertices ordered by this traversal.
\item Build a segment tree $\TB$ whose leaf nodes from left 
to right represent the vertices in list ${\cal L}$.
\item Augment each node $z$ of $\TB$ with a binary search tree ${\cal E}(z)$, storing all the edges $(u,v)\in E$ where
$u$ is a leaf node in the subtree rooted at $z$ in $\TB$. These edges are sorted according to the position of the second endpoint in ${\cal L}$.
%
\end{enumerate}

\begin{figure}[ht]
\centering
\includegraphics[width=.9\linewidth]{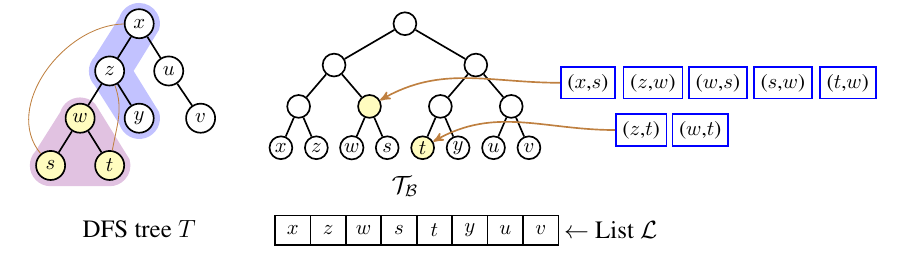}
\caption{(i) The highest edge from subtree $T(w)$ on $path(x,y)$ is edge $(x,s)$ 
and the lowest edges are edge $(z,w)$ and $(z,t)$.
(ii) The vertices of $T(w)$ are represented as union of two subtrees in segment tree $\TB$.}
\label{fig_xfast}
\end{figure}

The construction of $\DD$ described above ensures the following properties which are helpful in answering a query $Query(T(w),x,y)$ (see Figure \ref{fig_xfast}).
\begin{itemize}
\item $T(w)$ is present as an interval of vertices in $\cal L$ (by step 1). 
Moreover, this interval can be expressed as a union of $O(\log n)$ disjoint subtrees in $\TB$ (by step 2).
Let these subtrees be ${\TB}(z_1),\ldots,{\TB}(z_q)$.
\item It follows from the heavy-light decomposition used in step 1 that path $path(x,y)$ can be divided into $O(\log n)$ subpaths
 $path(x_1,y_1),\ldots,path(x_\ell,y_\ell)$ such that each subpath $path(x_i,y_i)$ is an interval in $\cal L$. 
\item 
Let query $Q(z,x,y)$ return the edge on $path(x,y)$ from the vertices in the subtree $\TB(z)$,  
that is closest to vertex $x$. 
Then it follows from step 3 that any query $Q(z_j,x_i,y_i)$ can be answered by a single predecessor or successor query on BST ${\cal E}(z_j)$ in $O(\log n)$ time.
\end{itemize}

To answer $Query(T(w),x,y)$, we thus find the edge closest to $x$ among all the edges reported by the queries
$\{ Q(z_j,x_i,y_i) | 1\le j \le q ~\mbox{and}~ 1\le i \le \ell \}$.
Thus, $Query(T(w),x,y)$ can be answered in $O(\log^3 n)$ time. 
Notice that $Query(w,x,y)$ can be considered as a special case 
where $q=1$ and $\TB(z_1)$ is the leaf node of $\TB$ representing $w$, i.e., $z_1=w$.
The space required by $\DD$ is $O(m\log n)$ as each edge is stored at $O(\log n)$ levels in $\TB$. 
Now, the segment tree $\TB$ can be built in linear time. Further, for every node $u\in \TB$, the sorted list of edges 
in ${\cal E}(u)$ can be computed in linear time by merging the sorted lists of its children. 
Thus, the binary search tree ${\cal E}(u)$ for each node $u\in \TB$ can be built in time
linear in the number of edges in ${\cal E}(u)$.
Hence the total time required to build this data structure is $O(m\log n)$. 
Thus, we have the following theorem.

\begin{theorem}
The queries $Query(T(w),x,y)$, $Query(w,x,y)$ on $T$ can be answered in $O(\log^3 n)$ 
worst case time using a data structure $\DD$ of size $O(m \log n)$, which  
can be built in $O(m\log n)$ time.
\label{thm:DS}
\end{theorem}

\noindent \textbf{Note: } 
Procedure~\ref{alg:re-root} can also use a simpler version of $\cal D$ 
which requires a smaller query time. However, our \textit{generic} algorithm 
(described in Section~\ref{sec:algo}) would require these additional features of $\cal D$ 
as follows.
\begin{enumerate}
\item For Procedure~\ref{alg:re-root}, the binary search tree ${\cal E}(u)$ stored at each node $u$
of $\TB$ can be 
replaced by an array storing the sorted list of edges, making it simpler to implement. 
However, our \textit{generic} algorithm also requires deletion of edges from $\DD$. An edge 
can be deleted from $\cal D$ by deleting the edge from the binary search trees
stored at its endpoints 
and their ancestors in $\TB$. Since a deletion in a binary search tree takes $O(\log n)$ time, 
an edge can be deleted from $\DD$ in $O(\log^2 n)$ time.
\item Procedure~\ref{alg:re-root} only performs the second type of query, i.e., $Query(T(w),x,y)$.
Thus, it would essentially be querying only the part of $path(x,y)$ comprising of the ancestors of $w$ in 
$path(x,y)$. This is thus equivalent to $Query(T(w), LCA(x,w),$ $LCA(y,w))$ (see Figure~\ref{fig_xfast}), 
which will answer the required query as the only edges from $T(w)$ in this interval are incident on $path(x,y)$. 
In such a case, heavy-light decomposition and hence division of $path(x,y)$ to $O(\log n)$ subpaths would not be required. 
Hence, on each node in $u\in\TB$, the query is performed for a single path, requiring total $O(\log^2 n)$ time. 
However, our \textit{generic} algorithm also uses the first type of query, i.e.,  $Query(w,x,y)$, 
where $w$ can be an ancestor of $x$ and $y$. In such a case, we need to perform the query 
only on contiguous intervals of ${\cal L}$ as the interval between $x$ and $y$ in $\cal L$ 
would have several other edges from $w$ that are not incident on $path(x,y)$. 
This necessitates the use of heavy-light decomposition and hence each query requires $O(\log^3 n)$
time.
\end{enumerate}

\section{Handling multiple updates - Overview}
\label{sec:Overview2}
DFS tree can be computed in $\tilde{O}(n)$ time after a single update in the graph, by reducing it to Procedure \ref{alg:re-root}.
However, the same procedure cannot be directly applied to handle a sequence of updates because of the following reason.
The efficiency of Procedure \ref{alg:re-root} crucially depends on the data structure $\cal D$ which is built using the DFS tree $T$ of the original graph.
Thus, when the DFS tree is updated, we are required to rebuild $\cal D$ for the updated tree.
Now, rebuilding $\cal D$ is highly inefficient because it requires $O(m\log n)$ time.  
Thus, in order to handle a sequence of updates, our aim is to use the same $\cal D$ for handling multiple updates,
without having to rebuild it after every update. We now give an overview of the algorithm that reports the DFS tree 
after a set $U$ of updates. 

%
%

%
%

In case of a single update, all the edges reported by $\cal D$ are added to the final DFS tree $T^*$.
However, while handling multiple updates, we use $\cal D$ to build {\em reduced adjacency lists} for vertices of the graph, such that the
DFS traversal of the graph using these {\em sparser} lists gives the DFS tree of the updated graph. 
Now, the data structure $\cal D$ finds the lowest/highest edge from a subtree of $T$ to an ancestor-descendant path of $T$.
Thus, in order to employ $\cal D$ to report DFS tree of $G+U$, we need to ensure that the queried subtrees and paths 
do not contain any failed edges or vertices from $U$. 
Hence, for any set $U$ of updates, we compute a partitioning of $T$ into 
a disjoint collection of ancestor-descendant paths and subtrees
such that none of these subtrees and paths contain any failed edge or vertex. 
An important property of this partitioning is that there are no edges from $G$ lying between
any two subtrees in this partitioning. We refer to this partitioning as a {\em disjoint tree partitioning}.
Note that this partitioning depends only upon  the vertex and edge failures present 
in the set $U$.

Recall that during the DFS traversal we need to find the lowest edge from each component $C$ of the unvisited graph.
It turns out that any component $C$ can be represented as a union of subtrees and ancestor-descendant paths of the original DFS tree $T$. 
The components property can now be employed to compute the reduced adjacency lists of the vertices of the graph as follows.
We just find the lowest edge from each of the subtrees and the ancestor-descendant paths to $T^*$ by querying the data structure $\DD$.
Let this edge be $(x,y)$ where $x\in T^*$ and $y\in C$. We can just add $y$ to the reduced adjacency list $L(x)$ of $x$.
Since the components property ensures the remaining edges to $T^*$ can be ignored, the DFS traversal would thus consider all possible 
candidates for the lowest edge from every component $C$ to $T^*$.
Let the initial disjoint tree partitioning consist of a set of ancestor-descendant paths ${\cal P}$ 
and a set of subtrees ${\cal T}$. The algorithm for computing a DFS tree of $G+U$ can be summarized as follows:

{\em Perform the static DFS traversal on the graph with the elements of ${\cal P}\cup{\cal T}$ as the {\em super} vertices.
Visiting a super vertex $v^*$ by the algorithm involves extracting an ancestor-descendant path $p_0$ 
from $v^*$ and attaching it to the partially grown DFS tree $T^*$. The remaining part of 
$v^*$ is added back to $\PP \cup \TT$ as {\em new} super vertices.
Thereafter, the reduced adjacency lists of the vertices on path $p_0$ are computed using the data structure
$\DD$. The algorithm then continues to find the next super vertex using the reduced adjacency lists and so on. }

\section{Disjoint Tree Partitioning}
\label{sec:partition}

We formally define disjoint tree partitioning as follows.


\begin{definition}
Given a DFS tree $T$ of an undirected graph $G$ and a set $U$ of failed vertices and edges,
let $A$ be a vertex set in $G+U$. The disjoint tree partitioning defined by $A$ 
is a partition of the subgraph of $T$ induced by $A$ into
\begin{enumerate}
\item A set of paths $\PP$ such that
(i) each path in $\PP$ is an ancestor-descendant path in $T$ and does not contain any deleted
edge or vertex, and (ii) $|\PP| \le |U|$.
\item A set of trees $\TT$ such that 
each tree $\tau \in \TT$ is a subtree of $T$ which does not contain any deleted
edge or vertex.
\end{enumerate}
Note that for any $\tau_1,\tau_2\in\TT$, there is no edge between $\tau_1$ and $\tau_2$ because 
$T$  is a DFS tree.
\label{definition:disjoint-tree-partition}
\end{definition}




The disjoint tree partitioning for set $A=V\setminus \{r\}$ can be computed as follows.
Let $V_f$ and $E_f$ respectively denote the set of failed vertices and edges associated with the updates $U$.
We initialize $\PP=\emptyset$ and $\TT=\{T(w)~|~w ~\mbox{is a child of $r$}\}$. We refine the
partitioning by processing each vertex $v\in V_f$ as follows 
(see Figure \ref{figure:disjoint-tree-partition} (i)). 
\begin{itemize}
\item 
If $v$ is present in some $T'\in \TT$, we add the path from $par(v)$
to the root of $T'$ to $\PP$. We remove $T'$ from $\TT$ and 
add all the subtrees hanging from 
this path to $\TT$. 
\item 
If $v$ is present in some path $p\in \PP$, we split $p$ at $v$ into 
two paths. We remove $p$ from $\PP$ and add these two paths to $\PP$.
\end{itemize}
Edge deletions are handled as follows.
We first remove edges from $E_f$ that don't appear in $T$. Processing of 
the remaining edges from $E_f$ is quite similar to the processing of $V_f$
as described above. For each edge $e\in E_f$, just visualize deleting an 
imaginary vertex lying at mid-point of the edge $e$ (see Figure \ref{figure:disjoint-tree-partition} (ii)).
It takes $O(n)$ time to process any $v\in V_f$ and any $e\in E_f$.  


\begin{figure}[th]
\centering
\includegraphics[width=\linewidth]{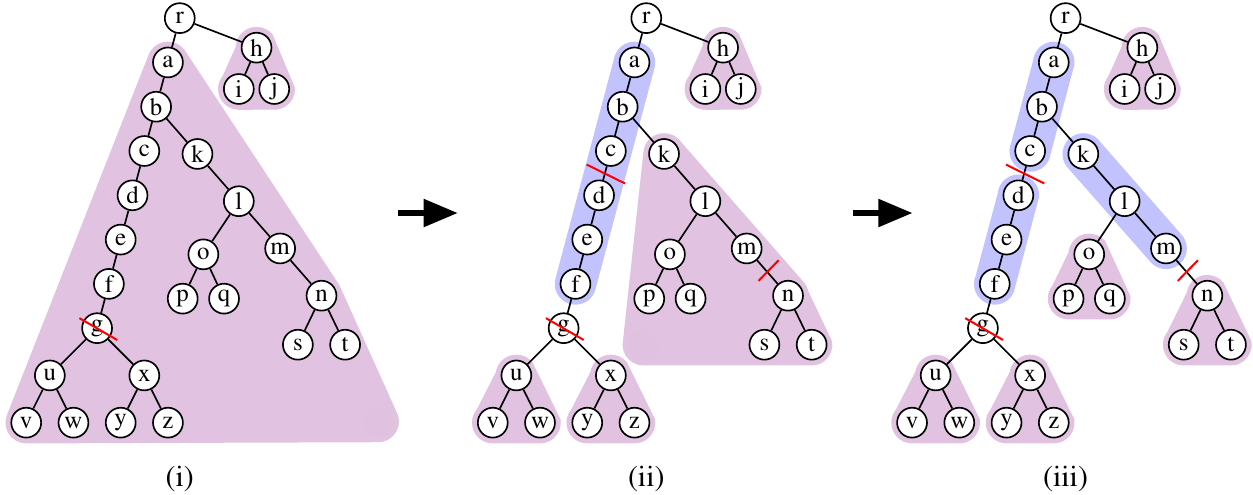}
\caption{
Disjoint tree partitioning for $V\setminus \{r\}$:
(i) Initializing $\TT=\{T(a),T(h)\}$ and $\PP=\emptyset$,
(ii) Disjoint tree partition obtained after deleting the vertex $g$.
(iii) Final disjoint tree partition obtained after deleting the edges $(c,d)$ and $(m,n)$.}
\label{figure:disjoint-tree-partition}
\end{figure}


Note that each update can add at most one path 
to $\PP$. So the size of $\PP$ is bounded by $|U|$. The fact that $T$ is a DFS tree of $G$ ensures that no two
subtrees in ${\cal T}$ will have an edge between them. So $\PP \cup {\cal T}$
satisfies all the conditions stated in Definition 
\ref{definition:disjoint-tree-partition}. 

\begin{lemma}
Given an undirected graph $G$ with a DFS tree $T$ and a set $U$ of failing vertices and edges, 
we can find a disjoint tree partition of set $V\setminus\{r\}$ 
in $O(n|U|)$ time. 
\label{lemma:partition}
\end{lemma}

\section{Fault tolerant DFS Tree}
\label{sec:algo}

We first present a fault tolerant algorithm 
for a DFS tree.
Let $U$ be any given set of failed vertices or edges in $G$. 
In order to compute the DFS tree $T^*$ for $G+U$, our algorithm first constructs a disjoint tree 
partition ($\TT,\PP$) for $V\backslash \{r\}$ defined by the updates $U$ 
(see Lemma \ref{lemma:partition}). Thereafter, it can be 
visualized as the static DFS traversal on the graph whose ({\em super}) 
vertices are the elements of $\PP\cup \TT$. Note that our notion of super 
vertices is for the sake of understanding only.

\begin{figure*}[t]
\centering
\includegraphics[width=\linewidth]{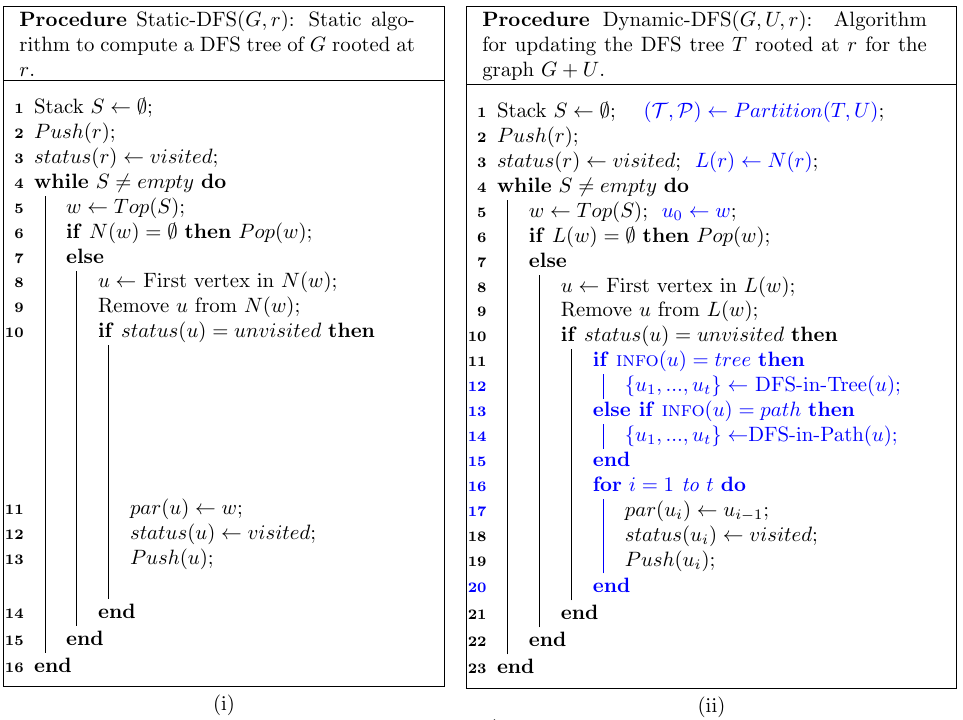}
\caption{The static (and dynamic) algorithm for computing (updating) 
a DFS tree. The key differences are shown in blue.}
\label{fig:DFSCompare}
\end{figure*}

Consider the stack-based implementation of the static algorithm for computing a DFS tree rooted at a vertex $r$ 
in graph $G$ (refer to Figure \ref{fig:DFSCompare}(i)). 
Our algorithm for computing DFS tree for $G+U$ (refer to 
Figure \ref{fig:DFSCompare}(ii)) is quite similar to the static algorithm.
The only points of difference are the following.
\begin{itemize} 
\item
In the static DFS algorithm whenever a vertex is visited, it is attached to the DFS tree
and pushed into the stack $S$.
In our algorithm when a vertex $u$ in some super vertex $v_s\in \PP\cup\TT$ is visited, 
a path starting from $u$ is extracted from $v_s$ and attached to the DFS tree, and  this entire path is pushed into the stack $S$. 
\item 
Instead of scanning the entire adjacency list $N(w)$ of a vertex $w$, 
the reduced adjacency list $L(w)$ is scanned. 
\end{itemize}

When a path is extracted from a super vertex $v_s$, the remaining unvisited part of $v_s$
is added back to $\TT\cup\PP$. 
However, we need to ensure that the properties of 
disjoint tree partitioning are satisfied in the updated $\TT\cup\PP$.
This is achieved using Procedure DFS-in-Path and Procedure DFS-in-Tree, which also build
the reduced adjacency list for the vertices on the path.
The construction of a sparse reduced adjacency list is 
inspired by the components property which can be \textit{adapted} 
in the context of our algorithm as follows.



\begin{lemma}[Adapted components property]
When a path $p$ is attached to the partially constructed DFS tree $T^*$ during the algorithm, 
for every edge $(x,y)$, where $x\in p$ and $y$ belongs to the unvisited graph
the following condition holds.
Either $y$ is added to $L(x)$ or $y'$ is added to $L(x')$ for some edge $(x',y')$ where $x'$ is a descendant (not necessarily proper) of $x$ in $p$ and $y'$ is connected to $y$ in the unvisited graph.
\end{lemma}

We now describe how the properties of disjoint tree partitioning and hence 
the adapted components property
is maintained by our algorithm when a vertex $v\in v_s$ is visited by the traversal.
\begin{enumerate}
\item Let $v_s=path(x,y)\in \PP$. 
Exploiting the flexibility of DFS, we traverse from $v$ to the farther end of $path(x,y)$.
Now, $path(x,y)$ is removed from $\PP$ and the untraversed part of $path(x,y)$ 
(with length at most half of $|path(x,y)|$) 
is added back to ${\cal P}$. 
We refer to this as  
{\em path halving}. This technique was also used by Aggarwal and Anderson \cite{AggarwalA88} in their parallel algorithm for computing DFS tree in undirected graphs. 
Notice that $|\cal P|$ remains unchanged or decreases by 1 after this step.
\item 
Let $v_s=\tau\in {\cal T}$. 
Exploiting the flexibility of a DFS traversal, we traverse the path from $v$ to the root of $\tau$, say $x$, and add
it to $T^*$. Thereafter, $\tau$ is removed from $\TT$ and all the subtrees hanging from this path 
are added to ${\cal T}$. 
Observe that every newly added subtree is also a subtree of the original DFS tree $T$. 
So the properties of disjoint tree partitioning are satisfied after this step as well.
\end{enumerate}

Let $path(v,x)$ be the path extracted from $v_s$. 
For each vertex $w$ in this newly added path, we compute $L(w)$ ensuring
the adapted components property
 as follows.
\begin{itemize}
\item[(i)] For 
each path $p\in {\cal P}$, 
among potentially many edges 
incident on $w$ from $p$, we just add any one edge. 
\item[(ii)] For each tree $\tau' \in {\cal T}$, we add at most one edge to $L$ as follows.
Among all edges incident on $\tau'$ from
$path(v,x)$, if $(w,z)$ is the edge such that $w$ is nearest to
$x$ on $path(v,x)$, 
then we add $z$ to $L(w)$. 
However, for the case $v_s\in \TT$, we have to consider only the newly added subtrees in $\TT$ for this step.
%
This is because the disjoint tree partitioning ensures the absence of edges between 
$v_s$ and any other tree in $\TT$.
\end{itemize}

Figure \ref{figure:overview-of-dfs-on-supernodes} provides an illustration of
how $\TT\cup\PP$ is updated when a super vertex in $\TT\cup\PP$ is visited.

\begin{figure}[t]
\centering
\includegraphics[width=.8\linewidth]{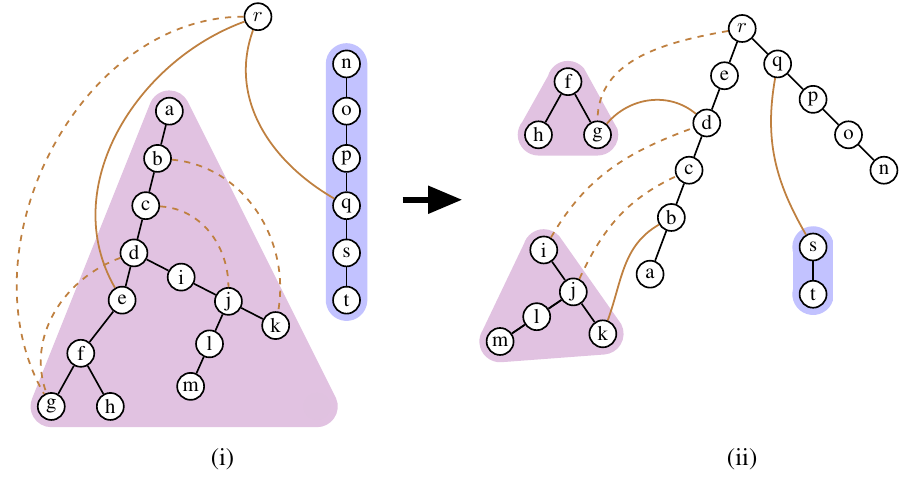}
\caption{
Visiting a super vertex from $\TT\cup\PP$.
(i) The algorithm visits $T(a)\in \TT$ using the edge $(r,e)$ and the $path(n,t)\in\PP$ using the edge $(r,q)$.
(ii) Traversal extracts $path(e,a)$ and $path(q,n)$ and augment it to $T^*$. 
The unvisited segments are added back to $\TT$ and $\PP$.
}
\label{figure:overview-of-dfs-on-supernodes}
\end{figure}


\subsection{Implementation of our Algorithm}
We now describe our algorithm in full detail. 
Firstly we delete all the failed edges in $U$ from the data structure $\DD$. 
Now, the algorithm begins
with a disjoint tree partition $(\TT,\PP)$ which evolves as the algorithm
proceeds. 
The state of any unvisited vertex in this partition is
captured by the following three variables.\\ 
-$\textsc{info}(u)$: this variable is set to $tree$ if $u$ belongs to a tree 
in $\TT$, and set to $path$ otherwise\\
-$\textsc{IsRoot}(v)$:~ this variable is set to $True$ if $v$ is the root of a 
tree in $\TT$, and $False$ otherwise. \\
-$\textsc{PathParam}(v)$:~ if $v$ belongs to some path, say $path(x,y)$, in 
$\PP$, then this variable stores the pair $(x,y)$, and $null$ otherwise.\\

\noindent\textbf{Procedure  Dynamic-DFS :}
For each vertex $v$, $status(v)$ is initially set as $unvisited$, and $L(v)$ 
is initialized to $\emptyset$. First a disjoint tree partition is computed for 
the DFS tree $T$ based on the updates $U$. The procedure Dynamic-DFS then
inserts the root vertex $r$ into the stack $S$. 
While the stack is non-empty, the procedure repeats the following 
steps. It reads the top vertex from the stack. Let this vertex be $w$. 
If $L(w)$ is empty then $w$ is popped out from the stack, else let $u$ be the 
first vertex in $L(w)$. If vertex $u$ is unvisited till now, 
then depending upon whether $u$ belongs to some tree in $\TT$ or some path in $\PP$, Procedure DFS-in-Tree or DFS-in-Path 
is executed.
A path $p_0$ is then returned to Procedure Dynamic-DFS where for each vertex of $p_0$
parent is assigned and status is marked visited. The whole of this path is then pushed into stack.
The procedure proceeds to the next iteration of While loop with the 
updated stack.\\

\begin{figure*}[t]
\centering
\includegraphics[width=\linewidth]{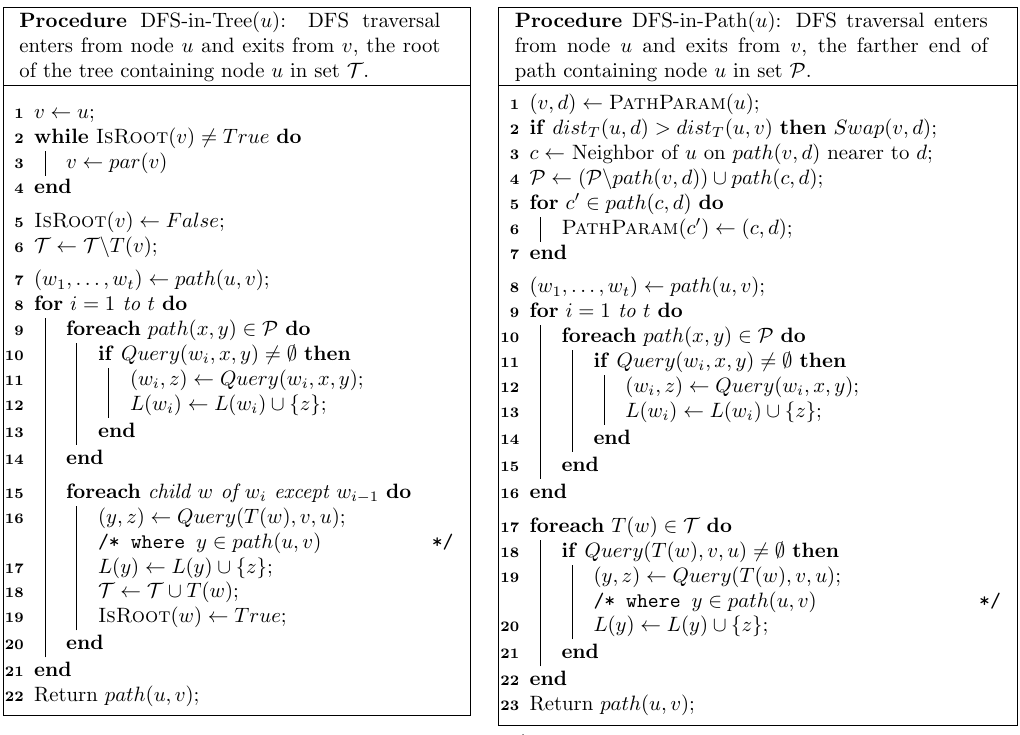}
\caption{The pseudocode of Procedures DFS-in-Tree and Procedures DFS-in-Path.}
\label{fig:treePathAlgo}
\end{figure*}


\noindent\textbf{ Procedure DFS-in-Tree :}
Let vertex $u$ be present in tree, say $T(v)$, in $\TT$ 
(the vertex $v$ can be found easily by scanning the ancestors of $u$ and 
checking their value of $\textsc{IsRoot}$).
The DFS traversal enters the tree from $u$ and leaves from the vertex $v$. 
Let $path(u,v)=\langle w_1=u,w_2\ldots,w_t=v \rangle$. The 
$path(u,v)$ is pushed into stack and attached to the partially constructed 
DFS tree $T^*$. We now update the partition $(\PP,\TT)$ and also update 
the reduced adjacency list for each $w_i$ present on $path(u,v)$ as follows. 
\begin{enumerate}
\item 
For each vertex $w_i$ and every path $path(x,y) \in \PP$, 
we perform $Query(w_i,x,y)$ on the data structure ${\cal D}$ that returns
an edge $(w_i,z)$ such that $z\in path(x,y)$. We add $z$ to $L(w_i)$. 
\item
Recall that since subtrees in $T$ do not have any cross edge between them, 
therefore, there cannot be any edge incident on $path(u,v)$ from trees which 
are already present in $\TT$. An edge can be incident only from the subtrees 
which were hanging from $path(u,v)$. $T(v)$ is removed from $\TT$ and 
all the subtrees of $T(v)$ hanging from $path(u,v)$ are inserted into $\TT$. 
For each such subtree, say $\tau$, inserted into $\TT$, we perform 
$Query(\tau,v,u)$ on the data structure ${\cal D}$ that returns an edge, 
say $(y,z)$, such that $z\in \tau$ and $y$ is nearest to $v$ on $path(u,v)$. 
We insert $z$ into $L(y)$.
\\
\end{enumerate}

\noindent\textbf{ Procedure DFS-in-Path  :}
Let vertex $u$ visited by the DFS traversal lies on a $path(v,y)\in \PP$. 
Assume $dist_T(u,v)>dist_T(u,y)$. The DFS traversal travels from $u$ to $v$
(the farther end of the path). The path $path(v,y)$ in set $\PP$ 
is replaced by its subpath that remains unvisited. The reduced adjacency list
of each $w\in path(u,v)$ is updated in a similar way as in the procedure 
DFS-in-Tree except that in step 2, we perform $Query(\tau,u,v)$ for
each $\tau \in \TT$.
%
Note that while performing step 1, 
the vertex $w_i$ can be an ancestor of the vertices of $path(x,y)$.
This is because the vertices of a path in $\PP$ can be ancestors of the vertices of 
another path in $\PP$. This was not true for Procedure~DFS-in-Trees because
vertices of a subtree in $\TT$ cannot be ancestors of vertices of any path in $\PP$.
Thus, our data structure $\cal D$ needs to support queries where $w_i$ is an 
ancestor of the queried path (refer to the note at the end of 
Section~\ref{sec:data-structure}).


The reader may refer to 
Figure \ref{fig:treePathAlgo} 
for pseudocode of Procedures DFS-in-Tree and DFS-in-Path. 
This completes the description of the
fault tolerant algorithm for DFS tree. This algorithm maintains 
the adapted components property
at each stage by construction given that the properties 
of disjoint tree partitioning are satisfied. 

\subsection{Correctness}
It can be seen that the following two invariants hold for the while loop in the 
Procedure Static-DFS described in Figure \ref{fig:DFSCompare} (i).
It is easy to see that 
these invariants imply the correctness of the algorithm, 
i.e., the generated tree is a rooted spanning tree where every non-tree edge is a back edge.

\begin{enumerate}
\item[$I_1$:] The sequence of vertices in the stack from bottom to top
constitutes an ancestor-descendant path from $r$ in the DFS tree computed.
\item[$I_2$:] For each vertex $v$ that is popped out,
all vertices in the set $N(v)$ have already been visited.
\end{enumerate}

These two invariants $I_1$ and $I_2$ also hold for Procedure 
Dynamic-DFS described in Figure \ref{fig:DFSCompare} (ii) as follows.
Invariant $I_1$ holds by construction as described in our algorithm.
Following lemma proves that invariant $I_2$ is maintained by our algorithm
since it follows the adapted components property
by construction.
%
\begin{lemma}
If the adapted components property 
is maintained by the Procedure Dynamic-DFS, then
invariant $I_2$ will hold true at each stage of the algorithm.
\label{lem:sufficiency}
\end{lemma}

\begin{proof} 
We give a proof by contradiction as follows.
Assume that $x$ is the first vertex that is popped out of the stack before some vertex $y\in N(x)$ is visited. 
Consider the time when a path $p$ containing $x$ was pushed in the stack. 
Clearly $y\notin L(x)$, hence using the adapted components property 
we know that some $y'\in L(x')$ is connected to $y$ in the unvisited
graph where $x'$ is a descendant (not necessarily proper) of $x$ in $p$. 
Let $p^*$ be a path between $y'$ and $y$ in the unvisited graph.
  
Now, consider the time when $x$ is popped out of the stack.
Clearly all its descendants including $x'$ have been popped out, so 
using invariant $I_2$ for $x'$, $y'$ has been visited by the traversal.
Thus, $p^*$ can be divided into two non-empty sets $A$ and $B$ denoting visited and
unvisited vertices of $p^*$ respectively. 
Here $y'\in A$ and $y\in B$, thus clearly for the last vertex of $p^*$ that is present in $A$, the invariant $I_2$ is not satisfied.
This contradicts our assumption that $x$ is the first vertex that is popped out of the stack for which $I_2$ is not satisfied. Thus, maintenance of the adapted components property 
 ensures the invariant $I_2$ in our algorithm.
\end{proof}

%
%
%
%

Hence, our algorithm indeed computes a valid DFS tree for $G+U$.

\subsection{Time complexity analysis}
As described earlier the disjoint tree partitioning and the components property play a key role in the 
efficiency of our algorithm. They allow us to limit the size of the reduced adjacency lists $L$,
that are built during the algorithm. 
Our algorithm computes $T^*$ by performing a DFS traversal on the reduced adjacency list $L$.
Thus, the time complexity of our algorithm is $O(n+|L|)$ excluding the time required to compute $L$. 

We first establish a bound on the size of $L$. 
In each step our algorithm extracts a path from $v_s\in \PP\cup\TT$ and attaches it to $T^*$.
Let $P_t$ and $P_p$ denote the set of such paths 
that originally belonged to 
some tree in $\TT$ and some path in $\PP$, respectively.
%
For every path $p_0\in P_t\cup P_p$ our algorithm performs the following queries on $\DD$.
\begin{enumerate}
\item[(i)] For each vertex $w$ in $p_0$, we query each path in $\PP$ for an edge incident on the vertex $w$. 
Thus, the total number of edges added to $L$ by these queries is $O(n|\PP|)$.

\item[(ii)] If $p_0$ belongs to $P_p$, then we query for an edge from each $\tau \in \TT$
to $p_0$. It follows from the path halving technique 
that each path in $\PP$ reduces to at most half of its length 
whenever some path is extracted from it and attached to $T^*$. 
 Hence, the size of $P_p$ is bounded by $|\PP| \log n$.
%

\item[(iii)] If $p_0$ belongs to $P_t$, then we query for an edge from only those
subtrees which were hanging from $p_0$. Note that these subtrees will now be added to
set $\TT$.
Hence, the total number of trees queried for 
this case will be bounded by number of trees inserted to $\TT$.
Since each subtree can be added to $\TT$ only once, these edges are bounded by $O(n)$
throughout the algorithm. 
\end{enumerate}

Thus, the size of $L$ is bounded by $O\big(n(1+|\PP|)\log n\big)$. Since each edge added to $L$ requires querying the data structure $\DD$ 
which takes $O(\log^3 n)$ time, the total time taken to compute $L$ is $O\big(n(1+|\PP|\log n)\log^3 n\big)$.
Thus, we have the following lemma.

\begin{lemma}
An undirected graph can be preprocessed to build a data structure of 
$O(m \log n)$ size such that for any set $U$ of $k$ failed vertices or 
edges (where $k \leq n$), the DFS tree of $G + U$ can be reported in 
$O(n(1+|\PP|\log n) \log^3 n)$ time.
\label{lemma:deletions-only}
\end{lemma}

From Definition \ref{definition:disjoint-tree-partition} we have that $|\PP|$ 
is bounded by $|U|$. Thus, we have the following theorem.

\begin{theorem}
An undirected graph can be preprocessed to build a data structure of 
$O(m \log n)$ size such that for any set $U$ of $k$ failed vertices or edges
(where $k \leq n$), the DFS tree of $G + U$ can be reported in 
$O(nk \log^4 n)$ time.
\label{theorem:deletions-only}
\end{theorem}
It can be observed that Theorem \ref{theorem:deletions-only} directly
implies a data structure for fault tolerant DFS tree.

\subsection{Extending the algorithm to handle insertions}
\label{sec:insert}
In order to update the DFS tree, our focus has been to restrict the number
of edges that are processed. For the case when the updates are deletions only, 
we have been able to restrict this number to $O(nk \log n)$, for a given set of $k$ updates 
(failure of vertices or edges). 
We now describe the procedure to handle vertex and edge insertions.
Let $V_I$ be the set of  vertices inserted, and $E_I$ be the set of edges inserted. 
(including the edges incident to the vertices in $V_I$).
If there are $k$ vertex insertions, the size of $E_I$ is bounded by $nk$. 
So even if we add all the edges in $E_I$ to the reduced adjacency lists, 
the size of $L$ would still be bounded by $O(nk \log n)$. 
Hence, we perform the following two additional steps before starting the DFS traversal.
\begin{itemize}
\item Initialize $L(v)$ to store the edges in $E_I$ instead of $\emptyset$. That is,
$L(v) \leftarrow \{y~|~ (y,v) \in E_I\}$ 
\item Each newly inserted vertex is treated as a 
singleton tree and added to $\TT$. That is, $\TT \leftarrow \TT \cup \{ x| x\in V_I\}$.
\end{itemize}

In order to establish that our algorithm, after incorporating the insertions, correctly computes a DFS tree
of $G+U$, we need to ensure that all the edges {\em essential} for DFS traversal as described in 
the adapted components property 
are added to $L$. 
%
All the essential edges from $G$  are added to $L$ during the algorithm itself. 
In case an essential edge belongs to $E_I$, the edge has already been added to $L$ during its initialization. 
Note that the time taken by our algorithm remains unchanged since the size of $L$ remains bounded by $O(nk\log n)$.
This completes the proof of our main result stated in Theorem \ref{main-result}. 


Let us consider the case when $U$ consists of insertions only. In this case $\PP$ will be an empty set. 
As discussed above, we initialize the reduced adjacency lists using $E_I$ whose size is equal to $|U|$.
Additionally, since the vertices in $V_I$ would be added to the set of trees, $|V_I|$ 
would be added to $n$. Hence, Lemma \ref{lemma:deletions-only} implies the following theorem.

\begin{theorem}
An undirected graph can be preprocessed to build a data structure of $O(m \log n)$ size such that for any set 
$U$ of $k$ vertex insertions and $m'$ edge insertions, a DFS tree of $G + U$ can be reported in 
$O(m'+(n+k)\log^3 n)$ time.
\label{theorem:insertions-only}
\end{theorem}

\noindent
\textbf{Note:} In Theorem~\ref{theorem:insertions-only}, the size of input is $k+m'$.
Also, even a single insertion may change $\Omega(n)$ edges of the DFS tree. 
Hence our algorithm is optimal upto $\tilde{O}(1)$ factors for processing edge or vertex
insertions if the DFS tree has to be maintained explicitly.

\section{Fully dynamic DFS}
\label{sec:fullyDyn}

We now describe the overlapped periodic rebuilding technique to convert
our algorithm for computing a DFS tree after $k$ updates to fully dynamic 
and incremental algorithms for maintaining a DFS tree. Similar technique was used by 
Thorup \cite{Thorup05} for maintaining fully dynamic all pairs shortest paths. 
 
In the fully dynamic model, we need to report 
the DFS tree after every update in the graph. 
Given the data structure $\DD$ built using the DFS tree of the graph $G$, 
we are able to report the DFS tree of $G+U$ after $|U|=k$ updates 
in $\tilde{O}(nk)$ time. This becomes inefficient if $k$ becomes large.
Rebuilding $\DD$ after every update is also inefficient as it takes $\tilde{O}(m)$ time to build $\DD$. 
Thus, it is better to rebuild $\DD$  after every $|U'|=c$ updates for a carefully chosen $c$.
Let $\DD'$ be the data structure built using the DFS tree of the updated graph $G+U'$ with $|U'|=c$.
$\DD'$ can thus be used to process the next $c$ updates efficiently (see Figure \ref{fig:replace} (a)).
The cost of building $\DD'$ can thus be amortized over these $c$ updates. 

To achieve an efficient worst case update time, we divide the building of $\DD'$ 
over the first $c$ updates. This $\DD'$ is then used by our algorithm in the next $c$ updates, during which 
a new $\DD''$ is built in a similar manner and so on (see Figure \ref{fig:replace} (b)). 
The following lemma describes how this technique can 
be used in general for any dynamic graph problem. For notational convenience we denote any function $f(m,n)$ as $f$.

\begin{lemma}
Let $D$ be a data structure that can be used to report the solution of a graph problem after a set of $U$ updates 
on an input graph $G$. If $D$ can be build in $O(f)$ time and the solution for graph $G+U$ can be reported in 
 $O(h+|U|\times g)$ time, then $D$ can be used to report the solution after every update in worst case
$O(\sqrt{fg}+h)$ update time, given that $\sqrt{f/g}\leq n$.
\end{lemma}

\begin{figure}[ht]
\centering
\includegraphics[width = \linewidth]{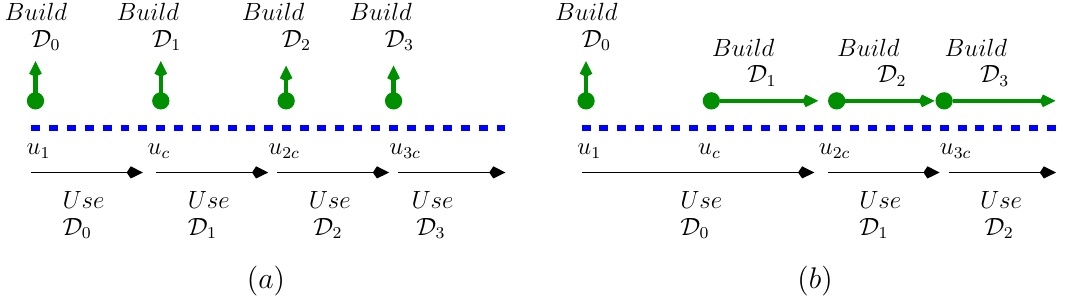}
\caption{
(a) Fully dynamic algorithm with amortized update time. 
(b) De-amortization of the algorithm.}
\label{fig:replace}
\end{figure}


\begin{proof}
We first present an algorithm that achieves amortized $O(\sqrt{fg}+h)$ 
update time. It is based on the simple idea of periodic rebuilding. Given 
the input graph $G_0$ we preprocess it to compute the data structure $D_0$ over
it. Now, let $u_1,...,u_c$ ($c\leq n$) be the sequence of first $c$ updates 
on $G_0$. To report the solution after $i^{th}$ update we use $D_0$ to compute 
the solution for $G_0+\{u_1,...,u_i\}$. This takes $O(h+(i\times g))$ time. 
So the total time for preprocessing and handling the first $c$ updates is 
$O(f+\sum_{i=1}^c h+(i\times g))$. Therefore, the average time for the first 
$c$ updates is $O(f/c+c\times g+h)$. Minimizing this quantity over $c$ gives
the optimal value $c_0=\sqrt{f/ g}$ which is bounded by $n$. So, after every
$c_0$ updates we rebuild our data structure and use it for the next $c_0$ updates 
(see Figure \ref{fig:replace}(a)). Substituting the value of $c_0$ gives the
amortized time complexity as $O(\sqrt{fg}+h~)$.

The above algorithm can be de-amortized as follows. Let $G_1,G_2,G_3,\ldots$ 
be the sequence of graphs obtained after $c_0,2c_0,3c_0,..$ updates. 
We use the data structure $D_0$ built during preprocessing to handle the first 
$2c_0$ updates. Also, after the first $c_0$ updates we start building the data 
structure $D_1$ over $G_1$. This $D_1$ is built in $c_0$ steps, thus the
extra time spent per update is $f/c_0=O(\sqrt{fg})$ only. 
We use $D_1$ to handle the next $c_0$ updates on graph $G_2$, and also in 
parallel compute the data structure $D_2$ over the graph $G_2$. 
(See Figure \ref{fig:replace}(b)). Since the time for
building each data structure is now divided in $c_0$ steps, we have that the
worst case update time as $O(\sqrt{fg}+h~)$.
\end{proof}


The above lemma combined with Theorems \ref{main-result} and 
\ref{theorem:insertions-only} directly implies the following results for the 
fully dynamic DFS tree problem and the incremental DFS tree problem, respectively.

(For the following theorem we use Theorem~\ref{main-result}, implying $f=m\log n$, $g=n\log^4 n$ and $h=0$.)
\begin{theorem}
There exists a fully dynamic algorithm for maintaining a DFS tree in an
undirected graph that uses $O(m\log n)$ preprocessing time and can report a 
DFS tree after each update in the worst case $O(\sqrt{mn} \log^{2.5} n)$ time.
An update in the graph can be insertion / deletion of an edge as well as a
vertex.
\end{theorem}

(For the following theorem we use Theorem~\ref{theorem:insertions-only}, implying 
$f=m\log n$, $g=\log^3 n$ and $h=n\log^3 n$.)
\begin{theorem}
There exists an incremental algorithm for maintaining a DFS tree in an
undirected graph that uses $O(m\log n)$ preprocessing time and can report a 
DFS tree after each edge insertion in the worst case $O(n \log^3 n)$ time.
\end{theorem}


\section{Applications}
\label{sec:appn}
Our fully dynamic algorithm for maintaining a DFS tree can  be used to solve 
various dynamic graph problems such as dynamic subgraph connectivity, biconnectivity and 2-edge connectivity.
Note that these problems are solved trivially using a DFS tree in the static setting.
Let us now describe the importance of our result in the light of the existing results for these problems.

\subsection{Existing Results}
The dynamic subgraph connectivity problem is defined as follows. Given
an undirected graph, the status of any vertex can be switched between  
{\it active} and {\it inactive} in an update. 
For any online sequence of updates interspersed with queries, 
the goal is to efficiently answer each connectivity queries 
on the subgraph induced by the active vertices. 
This problem can be solved by using dynamic connectivity
data structures \cite{EppsteinGIN97,Frederickson85,HolmLT01,KapronKM13} that answer
connectivity queries under an online sequence of edge updates. This is 
because switching the state of a vertex is equivalent to $O(n)$ edge 
updates. Chan \cite{Chan06} introduced this problem and showed that it 
can be solved more efficiently. He gave an algorithm using FMM (fast 
matrix multiplication) that achieves $O(m^{0.94})$ amortized update time 
and $\tilde{O}(m^{1/3})$ query time. Later Chan et al. \cite{ChanPR08} 
presented a new algorithm that improves the amortized update time to 
$\tilde{O}(m^{2/3})$. 
They also mentioned the following 
among the open problems.

\begin{enumerate}
\item
Is it possible to achieve constant query time with worst case sublinear ($o(m)$)
update time ?
\item
Can non trivial updates be obtained for richer queries such as counting
the number of connected components ?
\end{enumerate}

Duan \cite{Duan10} partially answered the first question affirmatively but at the
expense of a much higher update time and non-constant query time. He presented an 
algorithm with $O(m^{4/5})$ worst case update time
and $O(m^{1/5})$ query time, improving the worst case bounds for the problem.
Kapron et al. \cite{KapronKM13} presented a randomized algorithm for fully dynamic connectivity
which takes $\tilde{O}(1)$ time per update and answers the query correctly with high probability in $\tilde{O}(1)$ time, 
giving a Monte Carlo algorithm for subgraph connectivity with worst case $\tilde{O}(n)$ update time. 
Thus, their result answered the first question in a randomized setting.
However, in the deterministic setting both these questions were still open.
Our result answers both these questions affirmatively for the deterministic setting as well. 
Our fully dynamic algorithm directly provides an 
$\tilde{O}(\sqrt{mn})$ update time and $O(1)$ query time algorithm for 
the dynamic subgraph connectivity problem. Our algorithm maintains the
number of connected components simply as a byproduct. In fact, our fully dynamic
algorithm for DFS tree solves a generalization of dynamic subgraph connectivity 
- in addition to just switching the status of vertices, it allows insertion of 
new vertices as well. 
Hence the existing results offer different trade-offs between the update time and the query time, 
and differ on the types (amortized or worst case) of update time and the types (deterministic or randomized) of query time. 
Our algorithm, in particular, improves the deterministic worst 
case bounds for the problem (see Figure \ref{figure:subGraph}).
%
%
Further, unlike all the previous algorithms for dynamic subgraph connectivity, which use heavy 
machinery of existing dynamic algorithms, our algorithm is arguably much 
simpler and self contained. 


\begin{figure}[!ht]
\centering
\def\arraystretch{0.6}
\begin{tabular}{| m{4.5cm} | m{3.5cm} | m{2.2cm} |}
\hline
 {\bf References }							& {\bf Update Time} 					& {\bf Query Time} 		\\ \hline
Frederickson \cite{Frederickson85} (1985), 
Eppstein et. al \cite{EppsteinGIN97} (1997)	& $O(n\sqrt{n})$  					& $O(1)$ 	\\ \hline
Holm et al. \cite{HolmLT01} (2001)			& $\tilde{O}(n)$ amortized  			& $\tilde{O}(1)$ 		\\ \hline
Chan  \cite{Chan06} (2006) 					& $\tilde{O}(m^{0.94})$ amortized	& $\tilde{O}(m^{1/3})$ 	\\ \hline
Chan et al. \cite{ChanPR08} (2008) 			& $\tilde{O}(m^{2/3})$ amortized		& $\tilde{O}(m^{1/3})$ 	\\ \hline
Duan \cite{Duan10} (2010) 					& $\tilde{O}(m^{4/5})$ 				& $\tilde{O}(m^{1/5})$ 	\\ \hline
Kapron et al. \cite{KapronKM13} (2013)		& $\tilde{O}(n)$ 					& $\tilde{O}(1)$ \newline(Monte Carlo)\\ \hline
{\bf New\hfill } 								& $\tilde{O}(\sqrt{mn})$ 			& $O(1)$ 	\\ \hline
\end{tabular}
\caption{Current-state-of-the-art of the algorithms for the dynamic subgraph 
connectivity.} 
\label{figure:subGraph}
\end{figure}

Exploiting the rich structure of DFS trees, we also obtain 
$\tilde{O}(\sqrt{mn})$ update time algorithms for dynamic 
biconnectivity and dynamic 2-edge connectivity under vertex updates in a seamless 
manner. These problems have mainly been studied in the dynamic setting under edges updates.
Some of these results also allow insertion and deletion of isolated vertices. Our result, on the other
hand does not impose any such restriction on insertion or deletion of vertices.
Figure \ref{figure:subGraphB2E} illustrates our results and the existing results in the right perspective.
We now describe how our algorithm can be used to solve these problems.

\begin{figure}[h!]
\centering
\def\arraystretch{0.6}
\begin{tabular}{| m{4.5cm} | m{3.5cm} | m{2cm} | m{1.5cm} |}
\hline
{\bf References }							& {\bf Update Time} 					& {\bf Query Time} 		\\ \hline
Frederickson \cite{Frederickson85} (1985), 
Eppstein et. al \cite{EppsteinGIN97} (1997)\hfill$\dagger$& $O(n\sqrt{n})$  					& $O(1)$		\\ \hline
Henzinger \cite{Henzinger00} (2000)\hfill$*$				& $\tilde{O}(n\sqrt{n})$				& $O(1)$		\\	\hline
Holm et al. \cite{HolmLT01} (2001)\hfill$*\dagger$  		& $\tilde{O}(n)$ amortized  			& $\tilde{O}(1)$ \\ \hline
{\bf New\hfill $*\dagger$} 								& $\tilde{O}(\sqrt{mn})$ 			& $O(1)$ 	\\ \hline
\end{tabular}
\caption{Current-state-of-the-art of the algorithms for the dynamic biconnectivity $(*)$ and dynamic 2-edge connectivity $(\dagger)$ under vertex updates. }
\label{figure:subGraphB2E}
\end{figure}

\subsection{Algorithm}
The solution of dynamic subgraph connectivity follows seamlessly from our fully dynamic algorithm as follows.
As mentioned in Section \ref{sec:prelim}, we maintain a DFS tree rooted at a dummy vertex $r$, such that
the subtrees hanging from its children corresponds to the connected components of the graph. 
Hence, the connectivity query for any two vertices 
can be answered by comparing their ancestors at depth two (i.e. children of $r$).
This information can be stored for each vertex and updated whenever the DFS tree is updated.
Thus, we have a data structure for subgraph connectivity with worst case $\tilde{O}(\sqrt{mn})$ update time 
and $O(1)$ query time.
%
Our fully dynamic DFS algorithm can be extended to solve fully dynamic biconnectivity 
and 2-edge connectivity under 
vertex updates as follows. 

A set $S$ of vertices in a graph is called a {\em biconnected component} if it is a maximal set of vertices 
such that on failure of any vertex $w$ in $S$, the vertices of $S\setminus\{w\}$ 
remains connected. Similarly, a set $S$ is said to be {\em 2-edge connected component} if it is a
maximal set of vertices such that the failure of any edge with both endpoints in $S$
does not disconnect any two vertices in $S$. 
The biconnectivity and 2-edge connectivity queries can be answered easily by finding \emph{articulation points} and 
\emph{bridges} of the graph. It can be shown \cite{CormenLRS09} that two vertices 
belong to same biconnected component if and only if
the path connecting them in a DFS tree of the graph 
does not pass through any \emph{articulation point}. 
Similarly, two vertices belong to same 2-edge connected component if and only if 
the path connecting them in a DFS tree of the graph 
does not have a \emph{bridge}.
An articulation point and a bridge of a graph can be defined as follows:
\begin{definition}
Given a graph $G=(V,E)$, a vertex $v\in V$ is called an articulation point of $G$ 
if there exist a pair of vertices $x,y\in V$ such that every path between $x$ 
and $y$ in $G$ passes through $v$.
\end{definition}

\begin{definition}
Given a graph $G=(V,E)$, an edge $e\in E$ is called a bridge of $G$ if there exist
a pair of vertices $x,y\in V$ such that every path between $x$ and $y$ in $G$ 
passes through $e$.
\end{definition}

The articulation points and bridges of a graph can be easily computed by using DFS 
traversal of the graph. Given a DFS tree $T$ of an undirected graph $G$, we can 
index the vertices in the order they are visited by the DFS traversal. This index is 
called the \emph{DFN number} of the vertex. The \emph{high number} of a vertex $v$ 
is defined as the lowest DFN number vertex from which there is an edge incident to 
$T(v)$. Now, any non-root vertex $v$ will be an articulation point of the graph if 
high number of at least one of its children is equal to $DFN(v)$. The root $r$ of 
the DFS tree $T$ will be an articulation point if it has more than one child.
An edge $(x,y)$ of the DFS tree, where $x=par(y)$,  will be a bridge if 
the high number of $y$ is $DFN(x)$ and the high number of 
each child of $y$ (if any) is equal to $DFN(y)$. 
Thus,
given the high number of each vertex in the DFS tree, the articulation points and
bridges can be determined in $O(n)$ time.

We can augment our fully dynamic 
DFS algorithm with an additional procedure to compute high number of each vertex using the same
time bounds. 
For this we show that given any set of $k$ updates to graph $G$, 
while computing the new tree $T^*$ we also compute the high number of each vertex 
in $O(nk\log^4 n)$ time.
For each vertex $x$, let $a(x)$ denote the highest 
ancestor of $x$ in $T^*$ such that $(x,a(x))$ is an edge in $G+U$. Note that if 
$(x,a(x))$ is a newly added edge, then it can be easily computed by scanning all 
the new edges added to the graph. This is due to fact that the total number of new 
edges added to $G$ is bounded by $nk$. So we restrict ourselves to the case when 
$(x,a(x))$ was originally present in the graph $G$. 
Recall that our algorithm computes $T^*$ by attaching paths to the partially grown tree. 
Let $P_t$ and $P_p$ be the set of paths attached to $T^*$ (during its construction) that originally 
belonged to $\TT$ and $\PP$ respectively. Further, path halving ensures that the size of $P_p$ is 
bounded by $k \log n$.
For each path $p_0\in P_t\cup P_p$, let $H(p_0)$ denote the
vertex in $p_0$ that is closest to $r$ in $T^*$.


We now present the procedure for constructing a subset $A(x)$ of neighbors of $x$
while computing $T^*$ in $O(nk\log^4 n)$ time,  
such that the following condition holds.

\begin{itemize}
 \item For a vertex $x$, if $a(x)\notin A(x)$, then there is some descendant
 $y$ of $x$ in $T^*$ such that $a(x)\in A(y)$.
\end{itemize}

It is easy to see that if we get such an $A(x)$ for each $x$,
then high number of each vertex can be computed easily by processing the vertices of $T^*$
in bottom-up manner. 
Now, depending upon whether paths containing $x$ and $a(x)$ belong
to set $P_p$ or $P_t$, we can have different cases described as follows.


\begin{figure}[!ht]
\centering
\includegraphics[width=.95\textwidth]{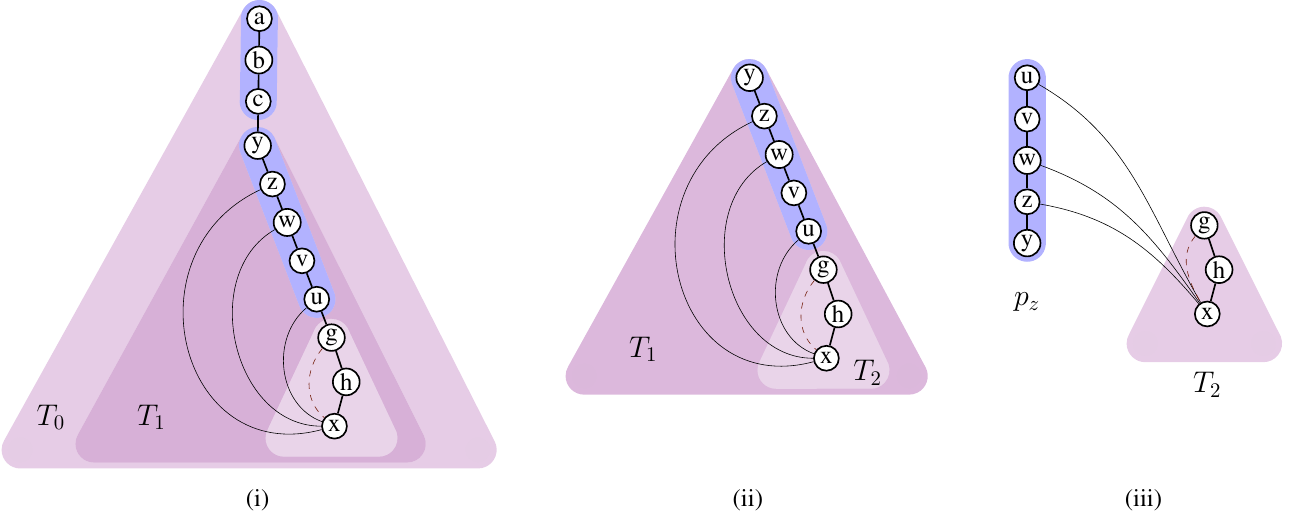}
\caption{(i) Before the beginning of algorithm vertex $x$ belongs to
tree $T_0\in \TT$, $z$ is the highest ancestor of $x$ in $T_0$ such that
$(x,z)$ is an edge.
(ii) The partitioning changes as the algorithm proceeds, $T_1(\in\TT)$ 
is the tree containing vertex $z$ just before it is attached to $T^*$.
(iii) A path containing vertex $z$ (i.e. $p_z$) is extracted from $T_1$
and attached to $T^*$. If $a(x)$ belongs to $T_0$, then it is the highest neighbor of $x$ in $p_z$.}
\label{figure:bi-connectivity}
\end{figure}

\begin{enumerate}
\item Vertex $a(x)$ lies on a path in $P_p$

For every vertex $v\in V$ and each path $p_0\in P_p$, we query 
$\cal D$ to compute the edge $(u,v)$ where $u$ is closest 
to $H(p_0)$ on path $p_0$, and add $u$ to $A(v)$. 
Note that if $a(x)$ lies on $p_0$, for $v=x$ the computed vertex 
$u$ will be same as $a(x)$. 

\item Vertex $x$ lies on a path in $P_p$

For each $u\in T^*$ and $p_0\in P_p$, we query $\cal D$ for an edge
$(u,y)$ such that the 
endpoint $y$ is farthest from $H(p_0)$ on path $p_0$. We add $u$ to $A(y)$. 
Now, consider a vertex $x$ on $p_0$ such that $a(x)=u$. If $x$ is equal to
$y$, then we have added $a(x)$ (i.e. $u$) to $A(x)$. If $x$ is not equal to
$y$, then we have added $a(x)$ (i.e. $u$) to $A(y)$ where $y$ is descendant 
of $x$ in $T^*$. 

\item Vertex $x$ and $a(x)$ lies on same path in $P_t$

For every vertex $v\in p_0$ for a path $p_0\in P_t$, 
we query $\cal D$ to compute the edge $(u,v)$ where $u$ is closest 
to $H(p_0)$ on path $p_0$, and add $u$ to $A(v)$. 
Note that for $x=v$, if $a(x)$ also lies on $p_0$, then $u$  
will be same as $a(x)$. 

\item Vertex $x$ and $a(x)$ lies on different paths in $P_t$

Let $x$ belong to $T_0$ in the initial disjoint tree partitioning $\TT\cup\PP$.
We claim that $a(x)$ would also belong to same tree $T_0$. 
This is because disjoint tree partitioning ensures the absence of edges between two subtrees in $\TT$. 
Let $z$ be the highest ancestor of $x$ in $T_0$ such that $(x,z)$ is an edge in $G+U$. 
Let $p_z$ be the path in $P_t$ containing vertex $z$. 

We now prove that $a(x)$ belongs to $p_z$. Recall that as the algorithm 
proceeds, our partitioning $\PP\cup \TT$ evolves with time.
Let $T_1$ be the tree in $\TT$ containing vertex $z$ just before $p_z$ is attached to $T^*$. 
Then $T_1$ is either same as $T_0$, or a subtree of $T_0$ (see Figure \ref{figure:bi-connectivity} (i)). 
Also, $a(x)$ must lie in tree $T_1$ since it cannot be an ancestor of $z$ in $T_0$. 
Now, let $T_2$ be the tree containing $x$ which is obtained on removal of $p_z$ from $T_1$.
Since $z$ is an ancestor of $x$ in $T_0$, the vertices in $T_2$ will eventually hang from some descendant of $z$ (not necessarily proper) in $T^*$.
For $a(x)$ to be the highest neighbor of $x$ in $T^*$, it should be an ancestor of $z$ in $T^*$, which is only possible if $a(x)\in p_z$.

Therefore, for each vertex $x$ belonging to a tree $T_0$ in $\TT$, we calculate
the highest ancestor $z$ of $x$ in $T_0$ such that $(x,z)$ is an edge in $G+U$.
We compute a list $l(z)$ that consist of all the vertices $x$ whose highest ancestor in $T_0$ is $z$.
Now, when $p_z$ is added to $T^*$, we process $l(z)$ as follows.
For every $v\in l(z)$, we query $\cal D$ for an edge $(u,v)$ 
where $u$ is closest to $H(p_z)$ on path 
$p_z$, and add $u$ to $A(v)$. Note that if $a(x)$ also lies in $T_0$, then $u$  
must be same as $a(x)$ (see Figure \ref{figure:bi-connectivity} (iii)). 

\end{enumerate}

Now, in the first two steps the total time taken is dominated by the number of 
queries between each path in $P_p$ and the vertices in $T$, i.e., $|P_p|\times n\times \log^3 n=O(nk\log^4n)$.
In the last two steps the total time taken is dominated by a single query for each vertex in $T$,
i.e., $n\times \log^3 n=O(n\log^3 n)$.
Thus, we have the following theorem.


\begin{theorem}
Given an undirected graph $G(V,E)$ with $|V|=n$ and $|E|=m$, we can maintain a
data structure for answering queries of biconnected components and 2 edge connectivity 
in a dynamic graph which takes $O(\sqrt{mn} \log^{2.5} n)$ update time, $O(1)$ query 
time and $O(m\log n)$ time for preprocessing.
\end{theorem}

\section{Lower Bounds}
\label{sec:LowerBounds}
We now prove two conditional lower bounds for maintaining a DFS tree under vertex or edge updates.

\subsection{Vertex Updates}
The lower bound for maintaining a DFS tree under vertex updates is based on Strong Exponential Time Hypothesis (SETH) as defined below:
\begin{definition}[SETH]
For every $\epsilon>0$, there exists a positive integer $k$, such that SAT on $k-$CNF formulas on $n$ variables cannot be solved in $\tilde{O}(2^{(1-\epsilon)n})$ time.
\end{definition}

Given an undirected graph $G$ on $n$ vertices and $m$ edges in a dynamic environment (incremental / decremental or fully dynamic) under vertex updates.
The status of any vertex can be switched between  {\it active} and {\it inactive} in an update. 
The goal of subgraph connectedness is to efficiently answer whether the subgraph induced by active vertices is connected.
Abboud and Williams\cite{AbboudW14} proved a conditional lower bound of $\Omega(n)$ per update based on SETH for answering dynamic subgraph connectedness queries. 
They proved that any algorithm for answering dynamic subgraph connectedness queries using arbitrary polynomial preprocessing time and $O(n^{1-\epsilon})$ amortized update time would essentially refute the SETH conjecture. They also proved that any algorithm for maintaining partially dynamic (incremental/decremental) subgraph connectedness using arbitrary polynomial preprocessing time and  $O(n^{1-\epsilon})$ worst case update time would essentially refute the SETH conjecture.

We present a reduction from subgraph connectedness to maintaining DFS tree under vertex updates requiring the 
algorithm to report whether the number of children of the root in any DFS tree of the subgraph is greater than 1. 
Thus, we establish the following:

\begin{theorem}
Given an undirected graph $G$ with $n$ vertices and $m$ edges undergoing vertex updates, an algorithm for maintaining DFS tree that can report the 
number of children of the root in the DFS tree with preprocessing time $p(m,n)$, update time $u(m,n)$ and query time $q(m,n)$ would imply an algorithm 
for subgraph connectedness with preprocessing time $p(m+n,n)$, update time $u(m+n,n)$ and query time $q(m+n,n)$.
\end{theorem}

\begin{proof}
Given the graph $G$ for which we need to query for subgraph connectedness, we make a graph $G'$ as follows.
We add all vertices and edges of $G$ to $G'$. Further, add another vertex $r$ called as \emph{pseudo root} and connect it to all other vertices of $G'$.
Thus, $G'$ has $n+1$ vertices and $m+n$ edges.
Now, in any DFS tree $T$ of $G'$ rooted at $r$,  the number of children of $r$ will be equal to the number of components in $G$.
Here subtrees rooted on each child of $s$ represents a component of $G$.
Any change on $G$ can be performed on $G'$ and query for subgraph connectedness in $G$ is equivalent to querying if $r$ has more than 1 child in $T$.
\end{proof}

Thus, any algorithm for maintaining fully dynamic DFS under vertex updates with arbitrary preprocessing time and $O(n^{1-\epsilon})$ amortized update time would refute SETH. Also, any algorithm for maintaining partially dynamic DFS under vertex updates with arbitrary preprocessing time and $O(n^{1-\epsilon})$ worst case update time would refute SETH. 

\subsection{Edge Updates}
We now present a lower bound for maintaining a DFS tree under edge updates that holds for any algorithm which maintains tree edges of the DFS tree explicitly.
In the following example we prove that there exists a graph $G$ and a sequence of edge updates $U$, such that any DFS tree of the graph would 
require a conversion of $\Omega(n)$ edges from tree edges to back edges and vice-versa after every pair of updates in $U$. 

\begin{figure}[ht]
	\hspace{1cm}\includegraphics[width=.8\linewidth]{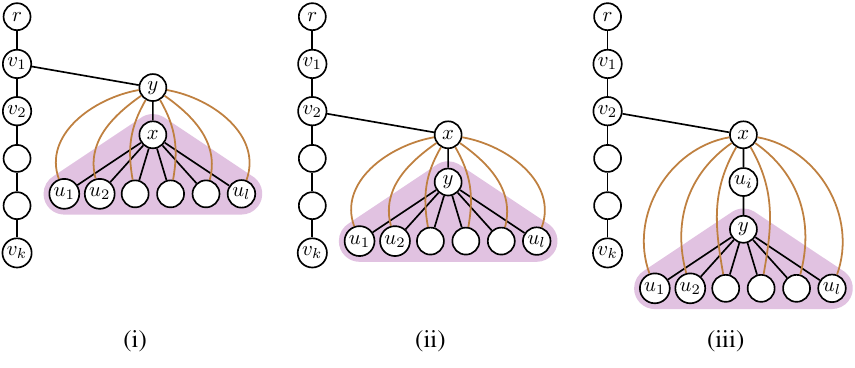}
	\caption{Worst Case Example for lower bound on maintaining DFS tree under fully dynamic edge updates.}

	\label{fig:worstCase}
\end{figure}

Consider the following graph for which a DFS tree rooted at $r$ is to be maintained under fully dynamic edge updates. 
There are $n/2$ vertices $u_1$,...,$u_l$ that have edges to vertices $x$ and $y$.
The remaining $n/2-3$ vertices $v_1$,..., $v_k$ are connected in form of a line as shown in Figure \ref{fig:worstCase}.
At any point of time one of $v_1,...,v_k$ (say $v_1$) is connected to either $x$ or $y$. 
The DFS tree for the graph is shown in Figure \ref{fig:worstCase} (i). Now, upon insertion of edge $(v_i,x)$ (say $i=2$) and deletion of edge $(v_1,y)$
the DFS tree will transform to either Figure \ref{fig:worstCase} (ii) or Figure \ref{fig:worstCase} (iii). Clearly $\Omega(n)$ edges are converted from tree edges to back edges and vice-versa. 
This can be repeated alternating between $x$ and $y$ ensuring that the new DFS tree 
requires $\Omega(n)$ after every two edge updates. 
Further, we repeat this for different $v_i$'s ensuring that the new DFS tree is not exactly 
the same as some previous DFS tree (thus memorization of the complete tree will not help).
Note that the same procedure can be applied to both the possible trees shown in Figure~\ref{fig:worstCase}(ii) and Figure~\ref{fig:worstCase}(iii).
Hence any algorithm maintaining tree edges explicitly takes $\Omega(n)$ time to handle such a pair of edge updates. 

\section{Conclusion}
\label{sec:Conclusion}
We have presented a fully dynamic algorithm for maintaining a DFS tree that takes worst case $\tilde{O}(\sqrt{mn})$ update time.
This is the first fully dynamic algorithm that achieves $o(m)$ update time.
In the fault tolerant setting our algorithm takes $\tilde{O}(nk)$ time to report a DFS tree, where $k$ is the number of 
vertex or edge failures in the graph. 
We show the immediate applications of 
fully dynamic DFS for solving various problems such as dynamic subgraph connectivity, biconnectivity and 2-edge connectivity.
We also prove the conditional lower bound of $\Omega(n)$ on maintaining DFS tree under vertex/edge updates.

	DFS tree has been extensively used for solving various graph problems in the static setting. 
Most of these problems are also solved efficiently in the dynamic environment.
However, their solutions have not used dynamic DFS tree. 
Furthermore, solutions to most dynamic graph problems under edge updates requires $o(n)$ update time.
However, this is not true for the vertex update variants of these problems.
In the light of $\Omega(n)$ lower bound for updating DFS under both edge and vertex updates, it becomes clear 
that dynamic DFS tree would be more applicable in dynamic graph problems under vertex updates.
The applications of our fully dynamic algorithm follows from the fact that it handles vertex updates 
which was not the case with the existing algorithms for maintaining DFS tree in any dynamic setting. 
This paper is thus an attempt to restore the glory of DFS trees for solving graph problems in the dynamic setting 
as was the case in the static setting.
We believe that our dynamic algorithm for DFS, on its own or after further
improvements/modifications, would encourage other researchers to use it in
solving various other dynamic graph problems.

\bibliography{paper}

\end{document}